%% file: main.tex
\documentclass[a4paper,UKenglish,cleveref, autoref, thm-restate]{lipics-v2021}

\hideLIPIcs  

\title{
    FO and MSO Model Checking on Temporal Graphs
}

\author{Michelle D\"oring}
    {Hasso Plattner Institute, University of Potsdam, Potsdam, Germany \and \url{https://www.notion.so/Michelle-D-ring-1dd43d6e8b7c800eadbdd32d73e21b72?pvs=25}}
    {michelle.doering@hpi.de}
    {https://orcid.org/0000-0001-7737-3903}
    {The project on which this report is based was funded by the Federal Ministry of Research, Technology and Space under the funding code “KI-Servicezentrum Berlin-Brandenburg” 16IS22092. Responsibility for the content of this publication remains with the author.}
    
\author{Jessica Enright}
    {School of Computing Science, University of Glasgow, UK}
    {jessica.enright@glasgow.ac.uk}
    {https://orcid.org/0000-0002-0266-3292}
    {}

\author{Laura Larios-Jones}
    {School of Computing Science, University of Glasgow, UK}
    {l.larios-jones.1@research.gla.ac.uk}
    {https://orcid.org/0000-0003-3322-0176}
    {} 
    
\author{George Skretas}
    {Hasso Plattner Institute, University of Potsdam, Potsdam, Germany}
    {georgios.skretas@hpi.de}
    {https://orcid.org/0000-0003-2514-8004}
    {}
\makeatletter
\newcommand{\DescLabel}[1]{\textcolor{gray}{\descriptionlabel{#1}}}
\makeatother

\authorrunning{M. Döring, J. Enright, L. Larios-Jones, G. Skretas} 

\Copyright{Jane Open Access and Joan R. Public} 



\category{} 

\relatedversion{} 

\nolinenumbers 

\EventEditors{John Q. Open and Joan R. Access}
\EventNoEds{2}
\EventLongTitle{42nd Conference on Very Important Topics (CVIT 2016)}
\EventShortTitle{CVIT 2016}
\EventAcronym{CVIT}
\EventYear{2016}
\EventDate{December 24--27, 2016}
\EventLocation{Little Whinging, United Kingdom}
\EventLogo{}
\SeriesVolume{42}
\ArticleNo{23}
\usepackage[dvipsnames]{xcolor}
\usepackage{xspace}
\usepackage{fancyhdr}
\usepackage{tcolorbox}
\usepackage{graphicx}
\usepackage{enumitem}

\usepackage{amssymb}
\usepackage{amsthm}
\usepackage{etoolbox}

\usepackage{amsmath,mathtools}
\usepackage{bbm}

\usepackage{stmaryrd} 
\usepackage{stackengine}
\usepackage[english,ruled,vlined,algo2e,linesnumbered]{algorithm2e}
\usepackage{bm}
\usepackage{thm-restate}

\usepackage{multirow}
\usepackage{tabularray}
\usepackage{tabularx}
\usepackage{colortbl}

\usepackage{anyfontsize}
\usepackage[percent]{overpic}
\usepackage{etoc}    
\usepackage{xparse}

\usepackage{todonotes}

\input{probenv}

\definecolor{darkgray}{gray}{0.25}
    \newenvironment{construction*}{%
        \par\vspace{0.25\baselineskip}
        \pushQED{\qed}%
        \noindent\textcolor{darkgray}{$\triangleright$}\;%
        \noindent\textbf{\textcolor{darkgray}{Construction~\theconstruction.}}\ %
    }{%
      \popQED\par\vspace{0.25\baselineskip}
    }
\theoremstyle{remark}

    \declaretheoremstyle[
      headfont=\bfseries,
      headformat=$\lozenge$\ \NAME\NOTE, 
      headpunct={.},
      notefont=\bfseries,
      bodyfont=\normalfont,
      spaceabove=6pt, spacebelow=6pt
    ]{cleancons}
    \theoremstyle{cleancons}
    
\newcommand{\ie}{i.\,e.,\xspace}
\newcommand{\eg}{e.\,g.,\xspace}

\newcommand{\scal}{\ensuremath{\mathcal{S}}\xspace}

\newcommand{\gcal}{\ensuremath{\mathcal{G}}\xspace}
\newcommand{\ecal}{\ensuremath{\mathcal{E}}\xspace}

\newcommand{\tuple}[1]{\ensuremath{\langle {#1} \rangle}\xspace}

\newcommand{\smallparagraph}[1]{\vspace{0.3em}\noindent\emph{#1}\quad}
\newcommand{\bigparagraph}[1]{\vspace{0.4em}\noindent\textbf{#1}}

\newcommand{\tg}{temporal graph\xspace}
\newcommand{\tgs}{temporal graphs\xspace}

\newcommand{\tdegree}{\ensuremath{\Delta^t}\xspace}
\newcommand{\sdegree}{\ensuremath{\Delta^s}\xspace} 

\newcommand{\lifetime}{\ensuremath{\Lambda}\xspace}

\newcommand{\NP}{\ensuremath{\mathtt{NP}}\xspace}

\newcommand{\paraNP}{\ensuremath{\mathtt{paraNP}}\xspace}

\newcommand{\Wone}{\ensuremath{\mathtt{W}[1]}\xspace}

\newcommand{\FPT}{\ensuremath{\mathtt{FPT}}\xspace}
\newcommand{\XP}{\ensuremath{\mathtt{XP}}\xspace}

\newcommand{\poly}{\ensuremath{\mathtt{poly}}\xspace}

\newcommand{\bigoh}{\mathcal{O}}

    \newcommand{\tw}{\ensuremath{\mathsf{tw}}\xspace}
    \newcommand{\pw}{\ensuremath{\mathsf{pw}}\xspace}

    \newcommand{\densef}{\ensuremath{\mathsf{d}\xspace}}
    \newcommand{\vim}{\ensuremath{\mathsf{vim}}\xspace}
    \newcommand{\vimText}{VIM\xspace}
    \newcommand{\tim}{\ensuremath{\mathsf{tim}}\xspace}
    \newcommand{\timText}{TIM\xspace}
    
    \newcommand{\twlifetime}{\ensuremath{\tw+\lifetime}\xspace}
    \newcommand{\twdegree}{\ensuremath{\tw+\tdegree}\xspace}

\newcommand{\MC}{\textsc{MC}\xspace}
\newcommand{\MCFO}{\textsc{FO MC}\xspace}
\newcommand{\MCMSO}{\textsc{MSO MC}\xspace}

\newcommand{\timbagTEXT}{TIM-bag\xspace}
\newcommand{\timbagTEXTs}{TIM-bags\xspace}
\newcommand{\vtimbag}
    {\ensuremath{\pi}\xspace}
\newcommand{\vtimbagset}
    {{\ensuremath{\Pi}}\xspace}

\newcommand{\inc}{\ensuremath{\mathsf{inc}}\xspace}
\newcommand{\edge}{\ensuremath{\mathsf{edge}}\xspace}
\newcommand{\Sedgesource}{\ensuremath{\mathsf{source}}\xspace}

\newcommand{\pres}{\ensuremath{\mathsf{pres}}\xspace}
\newcommand{\timebefore}{\ensuremath{<^T}\xspace}
\newcommand{\Sedgetarget}{\ensuremath{\mathsf{target}}\xspace}
\newcommand{\Sedgetime}{\ensuremath{\mathsf{pres}}\xspace}

\newcommand{\psuc}{\ensuremath{\mathsf{psuc}\xspace}}

\newcommand{\bagbefore}{\ensuremath{\mathsf{next}\xspace}}
\newcommand{\bag}{\ensuremath{\mathsf{bag}}\xspace}

\newcommand{\relGraph}[1]{\ensuremath{\lfloor #1 \rfloor}\xspace}
\newcommand{\relGraphLifetime}[1]{\ensuremath{\lfloor #1 \rfloor_{\lifetime}}\xspace}
\newcommand{\relGraphDegree}[1]{\ensuremath{\lfloor #1 \rfloor_{\tdegree}}\xspace}
\newcommand{\relGraphVim}[1]{\ensuremath{\lfloor #1 \rfloor_{\vim}}\xspace}
\newcommand{\relGraphTim}[1]{\ensuremath{\lfloor #1 \rfloor_{\tim}}\xspace}

\DeclareMathOperator{\gaifman}{Gf}
\newcommand{\eps}{\varepsilon}
\usepackage{pgf}
\makeatletter
\newdimen\@widthOfTo%
\newdimen\@widthOfLand%
\newdimen\@widthOfImplies%
\pgfmathsetmacro{\@scaleFactorImplies}{\@widthOfTo/\@widthOfImplies}%
\pgfmathsetmacro{\@scaleFactorTo}{\@widthOfLand/\@widthOfTo}%
    \renewcommand{\implies}{\mathrel{\raisebox{0.3ex}{\scalebox{\@scaleFactorImplies}{\ensuremath{\Longrightarrow}}}}}%
\pgfmathsetmacro{\@scaleFactorTo}{\@widthOfLand/\@widthOfTo}%
    \newcommand{\biimplies}{\mathrel{\raisebox{0.3ex}{\scalebox{\@scaleFactorImplies}{\ensuremath{\Longleftrightarrow}}}}}%
\makeatother

\input{_macros-mso-list}

\newif\iflong
\newif\ifshort
\longtrue 

\iflong
\else
\shorttrue
\fi

\begin{document}

\maketitle

\begin{abstract}
    Algorithmic meta-theorems provide an important tool for showing tractability
    of graph problems on graph classes defined by structural restrictions. While such results are well established for static graphs, corresponding frameworks for temporal graphs are comparatively limited.
    
    In this work, we revisit past applications of logical meta-theorems to temporal graphs and develop an extended unifying logical framework.
    Our first contribution is the introduction of logical encodings for the parameters vertex-interval-membership (VIM) width and tree-interval-membership (TIM) width, parameters which capture the signature of vertex and component activity over time. 
      Building on this, we extend existing monadic second-order (MSO) meta-theorems for bounded lifetime and temporal degree to the parameters VIM and TIM width, and establish novel first-order (FO) meta-theorems for all four parameters.
    
    Finally, we signpost a modular lexicon of reusable FO and MSO formulas for a broad range of temporal graph problems, and give an example. This lexicon allows new problems to be expressed compositionally and directly yields fixed-parameter tractability results across the four parameters we consider. \footnote{For the purpose of open access, the author(s) has applied a Creative Commons Attribution (CC BY) licence to any Author Accepted Manuscript version arising from this submission. No data were created or analysed in this work.}
    
\end{abstract}

\newpage
\section{Introduction}

\input{_intro-mi-2}

\subsection{Our contribution}
    
    We provide a unifying framework including the pre-existing encodings for lifetime (\lifetime) and temporal degree (\tdegree), and introduce novel encodings for the parameters \vim and \tim. These encodings are given in \Cref{sec:graph_encodings}.
    
    In \Cref{sec:temporal-meta}, we show that all four encodings admit the application of FO and MSO meta-theorems. Concretely, we analyse the \emph{Gaifman graph} induced by each encoding and prove that its structure is bounded by the respective parameter.
    This allows us to lift the static meta-theorems to temporal graphs, yielding the following temporal meta-theorems: any property expressible in FO/MSO logic over the corresponding encoding yields an \FPT algorithm parameterised by the respective parameter.
    \begin{restatable}{theorem}{temporalMSOmeta}
    \label{thm:temporal-mso-meta-short}
        \textsc{MSO Model Checking} on a temporal graph $\gcal=(G,\lambda)$ is \FPT when parameterised by (i)~$\twlifetime$, (ii)~$\twdegree$, (iii)~$\vim$, or (iv)~$\tim$, where \tw is the treewidth of $G$.
    \end{restatable}\vspace{-0.8em}
    \begin{restatable}{theorem}{temporalFOmeta}
    \label{thm:temporal-fo-meta-short}
        \textsc{FO Model Checking} on a temporal graph $\gcal=(G,\lambda)$ is \FPT when parameterised by (i)~$\lifetime$ if the footprint $G$ is nowhere dense, (ii)~$\tdegree$, (iii)~$\vim$, or (iv)~$\tim$.
    \end{restatable}
    
    Finally, 
    \iflong
    in \Cref{sec:cookbook}, 
    \fi
    we provide the \textit{Logic Cookbook}%
    \ifshort
     ~(see Section 5 in the full version):
    \fi a collection of FO and MSO formulas for over a dozen temporal graph problems studied in the literature.
    Besides serving as a reusable reference for future work, this section highlights the simplicity and modularity of this logical tool. 
    For example, in each of \cite{bumpus_edge_2023,cauvi_parameterized_2025,chakraborty_algorithms_2024,enright_families_2025,herrmann_timeline_2025}, a substantial contribution is an \FPT result proven via a specialised dynamic program and lengthy correctness arguments; in our cookbook, the same \FPT result follows from a compact formula combined with the meta-theorem. 
    
    \ifshort Some results and proofs are omitted due to space constraints and can be found in the full version of the paper appended at the end of the short version. \fi


\subsection{Related Work} \label{sec:related_work}

\input{_sec--related-work} 

\section{Preliminaries}
    A \textit{static graph} $G = (V, E)$ is a pair composed of a vertex set $V$ and an edge set $E$, where every element of the edge set is a pair of vertices, and two vertices in an edge are said to be \textit{adjacent}.  A \emph{clique} on $r$ vertices, denoted $K_{r}$ is a graph in which every pair of vertices is adjacent.

    \iflong
    \subsection{Temporal Graphs} \fi
    A \textit{\tg} $\gcal=(V,E,\lambda)$ consists of a static graph $G=(V,E)$, known as the \textit{footprint}, and a \textit{labelling} function $\lambda$. 
    The largest value of $\lambda$ is referred to as the \textit{lifetime} \lifetime.
    Alternatively, a temporal graph can be defined as a sequence $\mathcal{G} = \big( G_t = (V, E_t) \big)_{t \in \Lambda}$ of static \textit{snapshots}~$G_t$.
    \iflong The~two representations are equivalent via: $E_t := \{ e \in E : t \in \lambda(e) \}$ and $\lambda(e) := \{ t \in \Lambda : e \in E_t \}$. \fi   
    A~\tg is \emph{(un)directed} if the footprint is (un)directed.
    A pair $(e, t)$, where $e \in E$ and $t \in \lambda(e)$, is a \textit{temporal edge} with \textit{label} $t$. We denote the set of all temporal edges of \gcal by $\ecal$.
    A \textit{temporal path} is a sequence of temporal edges $\tuple{(e_i,t_i)}$ where $\tuple{e_i}$ forms a path in the footprint and the time labels $\tuple{t_i}$ are non-decreasing. If the time labels are strictly increasing, the path is \textit{strict}; otherwise, it is \textit{non-strict}. 
    If there is a temporal path from $u$ to $v$, we say \emph{$u$ reaches $v$}.
    A temporal graph is \textit{temporally connected} if for all $u,v\in V$, $u$ reaches $v$ and vice versa. 
    A temporal graph where reachability is considered exclusively using (non-)strict paths is called a \emph{(non-)strict \tg}.
    
    \iflong
    \subsection{Parameterised Complexity}
            We refer to the standard books for a basic overview of parameterised complexity theory~\cite{cygan_parameterized_2015,downey_fundamentals_2013,fomin_kernelization_2019}. At a high level, parameterised complexity studies the complexity of a problem with respect to its input size $n$  and the size of a parameter $k$.
            A problem is \emph{fixed-parameter tractable} (\FPT) by $k$ if it can be solved in time $f(k) \cdot \poly(n)$, where $f$ is a computable function.
            Showing that a problem is $\Wone$-hard parameterised by $k$ rules out the existence of such an \FPT algorithm under the assumption $\Wone \neq \FPT$.
            A less favourable, but still positive, outcome is an algorithm with an exponential running time $\bigoh(n^{f(k)})$ for some computable function~$f$; problems admitting such algorithms belong to the class~$\XP$.
            A problem is $\paraNP$-hard if it remains \NP-hard even when the parameter $k$ is constant. Thus, $\paraNP$-hardness excludes both \FPT and \XP algorithms under standard complexity assumptions.
        \fi

    \iflong \subsection{Parameters for Temporal Graphs}\fi
        We distinguish two families of parameters for temporal graphs: \emph{static} parameters which constrain the structure  of the (undirected) footprint and \emph{temporal} parameters which bound the temporal structure locally or the activity over time. See \Cref{fig:parameter-table} for an overview.

        \iflong\paragraph*{Static graph parameters}\fi
        We consider the maximum static degree $\sdegree$, the pathwidth $\pw$, the treewidth $\tw$, and (for graph classes) nowhere denseness.
        
        The \textit{static degree} of a vertex $v$ is defined as $\delta^s(v) = \lvert \{e\in E \colon v\in e\}\rvert$ and $\sdegree=\max_{v\in V} \delta(v)$ 
         denotes the \textit{maximum static degree} of \gcal.  \iflong
        
        \fi The \textit{treewidth} of a graph measures how structurally close it is to being a tree. Small treewidth means that the graph can be organised into small, overlapping pieces with a tree-like skeleton.
        \iflong Trees have treewidth~$1$, an $m\times n$ grid has treewidth $\max(m,n)$, and a complete graph on $k$ vertices has treewidth~$k-1$. \fi
        The treewidth of a \tg always refers to the \emph{undirected footprint}; for directed \tgs we take the treewidth of the underlying undirected graph of the directed footprint.
        \begin{definition}[Tree Decomposition, Treewidth~\cite{cygan_parameterized_2015}]
            Let $G=(V,E)$ be an undirected static graph.  
            A \emph{tree decomposition} of $G$ is a pair 
            $(T,\{B_u \colon u \in V(T)\})$ consisting of a tree $T$ and a family of bags $B_u \subseteq V$ such that
            \vspace{-0.5em}
            \begin{enumerate}[label=(\roman*)]
                \item $\bigcup_{u \in V(T)} B_u = V$,
                \item for every $e \in E$ there exists $u \in V(T)$ with $e \subseteq B_u$, and
                \item for every $v \in V$, the set $\{u \in V(T) \colon v \in B_u\}$ induces a connected subtree of~$T$.
            \end{enumerate}
            \vspace{-0.4em}
            The \emph{width} of $T$ is defined as $width(T) := \max_{u \in V(T)} \lvert B_u\rvert - 1$.
            The \emph{treewidth} of $G$ is 
            $
                \tw(G) := \min\{width(T) \colon T \text{ is a tree decomposition of } G\}.
            $
        \end{definition}
        
        The \textit{pathwidth} $\pw(G)$ is defined analogously to treewidth, with the restriction that the \emph{path decomposition} is required to be a path instead of a tree. Hence $\pw(G) \ge \tw(G)$ for all graphs $G$.
        A class of graphs has \textit{bounded treewidth} (respectively \textit{bounded pathwidth}) if there exists a constant $k$ such that every graph in the class satisfies $\tw(G)< k$ (resp. $\pw(G)<k$).
        \begin{table}[t]
            \begin{tabularx}{\linewidth}{c l l X}
                \hline
                Symbol & Name & Family & What it controls \\
                \hline
                $\sdegree$ & static degree & static & \# static edges locally \\
                $\pw$ & pathwidth & static & footprint path-likeness \\
                $\tw$ & treewidth & static & footprint tree-likeness \\
                $\densef$ & denseness function & static & sparsity measure \\
                \hline
                $\Lambda$ & lifetime & temporal & largest time label \\
                $\tdegree$ & temporal degree & temporal & \# temporal edges locally \\
                $\vim$ & vertex-interval membership & temporal & \# relevant vertices at a time\\
                $\tim$ & tree-interval membership & temporal & \# relevant, connected vertices at a time \\
                \hline
            \end{tabularx}
            \caption{Overview of the parameters considered in this paper.}
            \label{fig:parameter-table}
        \end{table}
        
        We use standard notions from topological minor theory; for details we refer to the book \cite{nesetril_sparsity_2012}. 
        Two graphs $G$ and $H$ are \emph{isomorphic} if there is a bijection between $V(G)$ and $V(H)$ that preserves adjacency.
        For $r\in\mathbb{N}$, an \textit{$r$-subdivision} of $H$ is obtained by selecting an arbitrary subset of edges and replacing each edge of that set by a path of length at most~$r + 1$ such that all paths are pairwise internally vertex-disjoint.
        The graph $H$ is a \textit{depth-$r$ topological minor}\footnote{In the literature, $\preceq^t_r$ is used to distinguish topological minors from classic minors; we will omit the $t$.} of a graph $G$, denoted $H\preceq_r G$, if some $r$-subdivision of $H$ is isomorphic to a subgraph of $G$.
        \begin{definition}[nowhere dense \cite{siebertz_nowhere_2016}]
            A graph class $\mathcal{C}$ is \emph{nowhere dense} if there exists a function $\densef\colon\mathbb{N}\to\mathbb{N}$ such that for every $r\in \mathbb{N}$ and every $G\in\mathcal{C}$ we have $K_{\densef(r)}\not\preceq_r G$.
            We call $\densef$ a \emph{denseness-function} of $\mathcal{C}$.
        \end{definition}
        Nowhere denseness of a graph class captures many common \textit{sparsity} notions on static graphs, including: bounded maximum degree, bounded treewidth/pathwidth, and planar graphs classes (classes for which we can draw graphs in a plane without intersecting edges).
        
        \iflong\paragraph*{Temporal graph parameters}\fi
        We consider the lifetime $\lifetime$, the maximum temporal degree $\tdegree$, the vertex-interval-membership width $\vim$, and the tree-interval-membership width $\tim$.

        The \textit{lifetime} \lifetime of a temporal graph is its largest time label. \iflong

        \fi The \textit{temporal degree} of a vertex $v$ at time $t$ is defined as $\delta^t(v) = \lvert \{(e,t) \colon v\in e,\; t\in\lambda(e)\}\rvert$ and $\tdegree=\max_{v\in V} \delta^t(v)$ denotes the \textit{maximum temporal degree} of~\gcal.\iflong 
        
        \fi The \emph{vertex-interval-membership (\vimText) width} and \emph{tree-interval-membership (\timText) width} capture the maximum number of vertices that are simultaneously `participating' in the temporal graph.
        Both measures use the \textit{activity-interval} of a vertex $v$, which starts at the first time at which $v$ has an incident edge and ends after the last time at which $v$ has an incident edge. For every snapshot outside this interval, $v$ does not interact with the rest of the graph.
        \newcommand{\alive}{alive\xspace}
        \begin{definition}[activity-interval]
            For $v\in V$, let $t_{min}(v)=\min \{\lambda(e) \colon e\in E, v\in e\}$ and $t_{max}(v)=\max \{\lambda(e) \colon e\in E, v\in e\}$. 
            The \emph{activity-interval of vertex} $v$ is defined as $A(v)=\left[t_{min}(v),t_{max}(v)\right]$.
            Similarly, we define the activity-interval of a static edge $e\in E$ as $A(e)=\left[\min\lambda(e), \max\lambda(e)\right]$.
            We say $v$ or $e$ is \emph{\alive} at time $t$ if $t\in A(v)$ or $t\in A(e)$, respectively.
        \end{definition}
        
        The \vimText decomposition consists of a path with one bag per time step containing exactly the vertices which are \alive at that time. 
        Note that every temporal graph has a unique \vimText decomposition and that a temporally connected graph has \vimText width exactly $n$. See \Cref{fig:def-vim-tim} left for an illustration of a \vimText decomposition.
        \begin{definition}[Vertex-Interval-Membership Width (Bumpus and Meeks \cite{bumpus_edge_2023})]
            The \emph{vertex-interval-membership (\vimText) decomposition} of a temporal graph \gcal is a sequence $(\vtimbag_t)_{t\in[\lifetime]}$ of \emph{bags} with $\vtimbag_t = \{v \in V : t\in A(v) \}$.
            The \emph{vertex-interval-membership width} of a temporal graph \gcal is $\mathsf{vim}(\gcal) = \max_{t \in [\lifetime]} \,\lvert \vtimbag_t\rvert$.
        \end{definition}
        
        The \timText width parameter generalises \vimText width from a decomposition over a path over time to a tree. This is motivated by the fact that for many temporal problems, vertices which are \alive at the same time but far from one another in the graph can be considered independently. The bags of a \timText decomposition have the property that any connected component of a snapshot is entirely contained in a bag at that time. Each bag at time $t$ is connected to the bags at time $t-1$ and $t+1$ with which it shares vertices. This forms a tree-shape (see \Cref{fig:def-vim-tim} right). 
        \begin{figure}[t]
            \centering
            \includegraphics[width=0.8\linewidth]{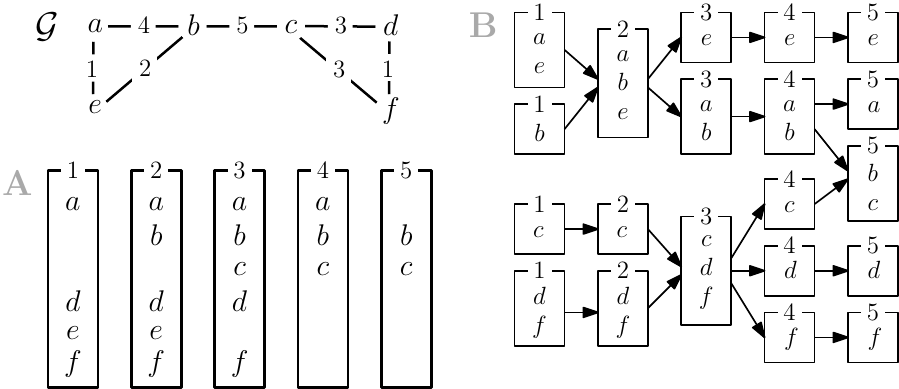}
            \caption{Example from~\cite{enright_families_2025} of a \vimText decomposition (\textcolor{darkgray}{A}) and a \timText decomposition (\textcolor{darkgray}{B}) of a temporal graph $\gcal$.}
            \label{fig:def-vim-tim}
        \end{figure}
        \begin{definition}[Tree-Interval-Membership Width \cite{enright_families_2025}]\label{def:TIMW}
            A \emph{tree-interval-membership (\timText) decomposition} of a temporal graph \gcal is a triple $(T,\vtimbagset,\tau)$ consisting of a labelled directed tree~$T$, a family of bags $\vtimbagset = \{\vtimbag_i \colon i\in V(T)\}$, and a node labelling function $\tau\colon V(T)\to [\Lambda]$ such that:
            \begin{enumerate}[label=(\roman*)]
                \item for every $v\in V(\gcal)$ and $t\in [\lifetime]$ there exists a unique $i\in V(T)$ with $\tau(i)=t$ and $v\in \vtimbag_i$,
                \item for every $(uv,t)\in\ecal(\gcal)$ there exists an $i\in V(T)$ with $\tau(i)=t$ and $\{u,v\}\subseteq \vtimbag_i$, and
                \item the directed edges of $T$ are $E(T)=\{(i,j) : \vtimbag_i\cap \vtimbag_j\neq \emptyset \text{ and } \tau(i)=\tau(j)+ 1\}$.
            \end{enumerate}    
            The \emph{width} of a \timText decomposition is defined as $width(T):=\max\{|\vtimbag_i|\colon i\in V(T)\}$.
            The \emph{\timText width} of a temporal graph $\mathcal{G}$ is $\tim(\gcal):=\min \{width(T)\colon T \text{ is a \timText decomposition of } \gcal\}$. 
        \end{definition}
        Since the VIM bags form a (not necessarily optimal) path decomposition of the undirected footprint, it follows that $\pw(\gcal) \le \vim(\gcal)$ (and hence $\tw(\gcal) \le \pw(\gcal) \le \vim(\gcal)$). Analogously, the TIM bags form a tree decomposition of the undirected footprint, and thus $\tw(\gcal) \le \tim(\gcal)$.
     
    \subsection{Logic in Static Graphs} \label{sec:static-logic}
        \textit{First-order logic\ifshort (FO)\fi} and \textit{monadic second-order logic\ifshort (MSO)\fi} provide a uniform language to express graph properties and, in turn, define whole graph classes by single descriptive sentences.
        \ifshort
        For a detailed treatment of FO and MSO, refer to \cite{courcelle_graph_2012,ebbinghaus_mathematical_1994,ebbinghaus_finite_1999}.
        
            A graph in the context of logic is encoded as a \emph{relational structure}: An undirected graph $G=(V,E)$ has the universe $U=V\cup E$ together with the incidence relation $\inc\subseteq E\times V$ such that $\inc(uv,x)$ for $x=u$ or $x=v$.
            A directed graph $D=(V,A)$ has universe $U=V\cup A$ with the relations $\Sedgesource,\Sedgetarget\subseteq A\times V$ such that $\Sedgesource((u,v),x)$ for $x=u$ and $\Sedgetarget((u,v),x)$ for $x=v$.
            
            An \textit{FO formula} uses element variables $x,y,\dots$ ranging over the universe $U$ and builds formulas from atomic predicates ($x=y$, $R(x_1,x_2,\dots)$) via the logical connectives $\neg, \wedge, \vee, \implies, \biimplies$, the existential quantifier $\exists$ and the universal quantifier $\forall$.
            An \textit{MSO formula} is an FO formula which can additionally use \emph{set variables} $X,Y,\dots\subseteq U$, the membership relation $x\in X$, and is allowed to quantify over sets.
            The algorithmic task of deciding whether a  graph satisfies a specific formula is called the \textit{model checking} problem. 
            
            \begin{problem}{Model Checking (MC)}
                \Input &A static graph $G$ and a logical formulas $\varphi$.\\
                \Prob & {Does $G$ satisfy $\varphi$?}
            \end{problem}
        \else
        
            These logics are interpreted over \emph{relational structures}. A relational structure $\mathcal{S}$ consists of a finite universe $U$ together with a finite set of relations $\mathfrak{R}$ on $U$. Before formally introducing FO and MSO, we specify how a static graph $G=(V,E)$ is represented as such a structure.
            \begin{definition} \label{def:static_encoding}
                Let $G=(V,E)$ and $D=(V,A)$ be an undirected and a directed static graph, respectively. The undirected relational structure \relGraph{G} is defined as: \vspace{-0.5em}
                \begin{description}
                    \item[universe:] $U=V\cup E$
                    \item[unary relations:] $V\subseteq U$, $E\subseteq U$
                    \item[relations:] \hphantom{.}
                    \vspace{-0.4em}
                    \begin{itemize}
                        \item $\inc\subseteq E\times V$, where $\inc(e,v) \Leftrightarrow v\in e$.
                    \end{itemize}\vspace{-0.4em}
                \end{description}\vspace{-0.4em}
                In the directed relational structure \relGraph{D}, $E$ is replaced by $A$ in the universe and unary relations, and $\inc$ is replaced by (i) $\Sedgesource\subseteq A\times V$, where $\Sedgesource((x,y),v) \Leftrightarrow v=x$,
                and (ii) $\Sedgetarget\subseteq A\times V$, where $\Sedgetarget((x,y),v) \Leftrightarrow v=y$. 
            \end{definition}
            \begin{definition}
            \label{def:mso}
                Let $\mathcal{S}$ be a relational structure with universe $U$ and relations $\mathfrak{R}$.
                
                Formulas in \emph{monadic second-order logic (MSO)} use two types of \emph{variables}: first-order variables $x,y,z,\dots$ ranging over elements of $U$, and second-order variables $X,Y,Z,\dots$ ranging over subsets of $U$.
                An \emph{MSO formula} is built inductively as follows: 
                \vspace{-0.4em}
                \begin{description}
                    \item[atomic formulas:]\hphantom{.}
                    \vspace{-0.4em}
                    \begin{itemize}
                        \item $x=y$ or $X=Y$ (equality);
                        \item $R(x_1,\dots,x_k)$ for $R \in \mathfrak{R}$;
                        \item $x \in X$ for a set variable $X$ and element $x$.
                    \end{itemize} 
                    \item[formulas:] If $\varphi$ and $\psi$ are MSO formulas, then so are: \vspace{-0.4em}
                    \begin{itemize}
                        \item $\neg \varphi$, $\varphi \vee \psi$, $\varphi \wedge \psi$, $\varphi \to \psi$;
                        \item $\exists x (\,\varphi(x)\,)$, $\forall x (\, \varphi(x)\,)$ (first-order quantification);
                        \item $\exists X (\, (\varphi(X)\,)$, $\forall X (\, \varphi(X)\,)$ (second-order quantification over sets).
                    \end{itemize}
                \end{description}
                \emph{First order logic (FO)} formulas neither contain second-order variables nor allow second-order quantification over sets. Note that every FO formula is also an MSO formula. 
            \end{definition}
            For more details on FO and MSO, refer to \cite{courcelle_graph_2012,ebbinghaus_mathematical_1994,ebbinghaus_finite_1999}.

        The algorithmic task of deciding whether a given graph satisfies a specific formula is the\textit{ model checking} problem. 
        
        \begin{problem}{Model Checking (MC)}
            \Input &A static graph $G$ and a logical formulas $\varphi$.\\
            \Prob & {Does $G$ satisfy $\varphi$?}
        \end{problem}
        \fi
        
        \noindent
        For FO and MSO there exist general meta-theorems that characterise the static graph classes on which their respective \MC problem is tractable.
        \begin{theorem}[Grohe--Kreutzer--Siebertz: \MCFO on nowhere dense \cite{grohe_deciding_2017}]
        \label{thm:FO-MC}
            Let $\mathcal{C}$ be a graph class with denseness function $\densef\colon\mathbb{N}\to\mathbb{N}$,
            i.e., $K_{\densef(r)}\not\preceq_r G$ for all $r\in\mathbb{N}$ and all $G\in\mathcal{C}$.
            Then for every $\varepsilon>0$ there exists a computable function $f$ and a computable radius function
            $r(\cdot)$ such that, given a graph $G\in\mathcal{C}$ and an FO formula $\varphi$,
            one can decide whether $G$ satifies $\varphi$ in time
            \ifshort $F\bigl(|\varphi|,\varepsilon,\densef(r(\varphi))\bigr)\cdot |V(G)|^{1+\varepsilon}$.
            \else
            \[
            F\bigl(|\varphi|,\varepsilon,\densef(r(\varphi))\bigr)\cdot |V(G)|^{1+\varepsilon}.
            \]
            (In particular, the dependence on the class $\mathcal{C}$ appears only through the values of $\densef$ at radii
            bounded in terms of the formula $\varphi$.)\fi
        \end{theorem}
            
        \begin{theorem}[Courcelle: \MCMSO on bounded treewidth, \cite{courcelle_graph_2012}]\label{thm:MSO-MC}
            There exists a computable function~$f$ such that, given a graph $G$ of treewidth $\tw$ and an MSO formula $\varphi$, one can decide whether
            $G$ satisfies $\varphi$ in time $f(\lvert\varphi\rvert,\tw)\cdot (|V|+|E|)$.
        \end{theorem}
            
        \iflong
        For optimization problems, there is an optimization variant of Courcelle’s meta-theorem which additionally requires a function $\alpha$ over which the property is optimised. Note that a variable is \emph{free} if it is not in the scope of a quantifier.
        \begin{theorem}[Arnborg, Lagergren, Seese: Extended MSO MC on bounded treewidth, \cite{arnborg_easy_1991}] \label{thm:MSO-MC-optimization}
            There exists an algorithm that, given an MSO formula $\varphi$ with free monadic variables $X_1,\dots,X_r$, an affine function $\alpha(x_1,,\dots,x_r)$, and a graph $G$ of treewidth \tw, computes the minimum (or maximum) of $\alpha(\lvert X_1\rvert, \dots, \lvert X_r\rvert)$ over all evaluations of $X_1,\dots,X_r$ that satisfy $\varphi$ on $G$, in time $f(\lvert \varphi\rvert,\tw) \cdot (|V|+|E|)$, where $f$ is a computable function.
        \end{theorem}

        In addition, Courcelle's theorem can be extended to find the \emph{number} of satisfying assignments of an MSO formula. This allows us to solve enumeration problems consisting of a property expressible in MSO. 
        \begin{theorem}[Courcelle: Counting MSO MC on bounded treewdith, \cite{courcelle_graph_2012}]
            Let $\psi(X_1,\ldots,X_j,x_1,\ldots,x_{\ell})$ be an MSO-formula with set variables $X_1,\ldots,X_j$ and individual variables $x_1,\ldots,x_{\ell}$. Let $\relGraph{G}$ be a relational structure with universe $U$. Given a tree decomposition of $\relGraph{G}$ of width \tw, the cardinality of the set $\psi(\relGraph{G})$ can be computed in time $f(\tw,|\psi|)\cdot |\relGraph{G}|$ for a computable function $f$.
        \end{theorem}
        \fi
        
\section{Temporal Graph Encodings}
\label{sec:graph_encodings}
    \ifshort
    To apply logical meta-theorems to \tgs, we encode them as relational structures that extend the standard static graph encoding with time information. Different encodings make different aspects of the temporal structure explicit and directly influence the structure of the resulting representation.
    We present four relational encodings—two known and two novel—each tailored to a specific notion of static or temporal sparsity. These encodings serve as the foundation for the temporal meta-theorems in \Cref{sec:temporal-meta}, yielding tractability of FO and MSO model checking under bounded
    (i)~$\lifetime$, (ii)~$\tdegree$, (iii)~$\vim$, and (iv)~$\tim$, possibly combined with structural static parameters.
    All encodings are presented for strict, undirected temporal graphs.
    Directed variants can be obtained by replacing the incidence relation with source and target relations. Non-strict reachability is obtained via adapting the successor relation in the degree encoding, or on the formula level for the lifetime, VIM, and TIM encoding.
    \else
    To apply logical meta-theorems to \tgs, we encode them as relational structures resembling the static graph encoding with additional structure to represent time.
    Different encodings make different aspects of the temporal structure explicit and, crucially, affect the structure of the resulting representation.
    
    We present four relational encodings---two known and two novel---each tailored to a specific notion of static and temporal sparsity. These encodings serve as the foundation for the temporal meta-theorems in \Cref{sec:temporal-meta}, yielding tractability of FO/MSO model checking under bounded (i)~$\lifetime$, (ii)~$\tdegree$, (iii)~$\vim$, and (iv)~$\tim$ possibly combined with structural static parameters.

    \bigparagraph{(Non)strict, (un)directed variants.}\quad
    We present all encodings explicitly for \emph{strict, undirected} \tgs.
    Encodings for \textit{directed} temporal graphs can be obtained by replacing the incidence relation with source and target relations
    $\Sedgesource,\Sedgetarget \subseteq E \times V$, where $\Sedgesource(e,v)$ (resp., $\Sedgetarget(e,v)$) holds if and only if $v$ is the tail (resp., head) of $e$.
    The distinction between strict and non-strict temporal reachability is handled as follows. In the degree encoding (\Cref{def:degree_encoding}), strictness is enforced at the structural level via the possible-successor relation (using $t' < t$ for strict and $t' \le t$ for non-strict reachability). In all other encodings, strictness is enforced at the formula level (e.\,g., when defining a logical formula for temporal paths).\fi

    \medskip
    \iflong\paragraph*{Lifetime encoding}\fi
    \label{subsec:lifetime_encoding}
        The lifetime encoding represents a temporal graph as a time-labelled static graph, mirroring the standard mathematical definition.
        The universe contains the vertices, temporal edges, and time steps. An incidence relation identifies the endpoints of each temporal edge, while a presence relation records when a temporal edge is present. A total linear order on the time steps encodes the global timeline.
        \begin{definition}[lifetime encoding]
        \label{def:lifetime_encoding}
            Let $\gcal=((V,E),\lambda)=(V,\mathcal{E})$ be an undirected strict \tg. The \emph{lifetime encoding} \relGraphLifetime{\gcal} is the relational structure defined as:
            \vspace{0.4em}

            \noindent
            \begin{minipage}[t]{0.2\linewidth}
            \DescLabel{universe}
            \end{minipage}%
            \hspace{0.5em}
            \begin{minipage}[t]{0.8\linewidth}
            $U=V\cup \mathcal{E}\cup L$
            \end{minipage}

            \noindent
            \begin{minipage}[t]{0.2\linewidth}
            \DescLabel{unary relations}
            \end{minipage}%
            \hspace{0.5em}
            \begin{minipage}[t]{0.8\linewidth}
            $V(\cdot),\ecal(\cdot), L(\cdot)$
            \end{minipage}

            \noindent
            \begin{minipage}[t]{0.2\linewidth}
            \DescLabel{binary relations}
            \end{minipage}%
            \hspace{0.5em}
            \begin{minipage}[t]{0.8\linewidth}
            \begin{itemize}[leftmargin=*, topsep=0pt, itemsep=0pt, parsep=0pt, partopsep=0pt]
                \item $\inc\subseteq \mathcal{E}\times V$
                where $\inc((uv,t),x)\Leftrightarrow x=u \text{ or } x=v$,
                \item $\mathsf{pres}\subseteq \mathcal{E}\times L$
                where $\mathsf{pres}((uv,t),t')\Leftrightarrow t'=t$,
                \item $\timebefore\subseteq L\times L$
                where $t_1\timebefore t_2\Leftrightarrow t_1<t_2$.
            \end{itemize}
            \end{minipage}
        \end{definition}
    \iflong
    \begin{figure}[h]
        \centering
        \includegraphics[width=0.75\linewidth]{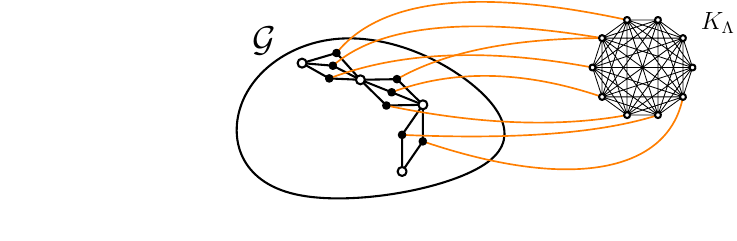}
        \caption{
            Illustration of the lifetime encoding $\relGraphLifetime{\gcal}$ of a temporal graph $\gcal$.
            Each static edge between two vertices is subdivided by temporal-edge-nodes (filled circles), which are connected via the presence relation ($\mathsf{pres}$, shown in orange) to the time-node $t \in L$ for which $\eps=(e,t)\in\ecal$.
            The total order $\timebefore$ on the time-nodes induces a clique $K_\Lambda$ of size \lifetime.}
        \label{fig:gaifman-lifetime_encoding}
    \end{figure}
    \fi
    
    \iflong\paragraph*{Degree encoding}\fi
    \label{subsec:degree_encoding}
        The degree encoding abstracts away from the concrete time labels of edges and instead focuses on the \emph{relative order} of neighbouring temporal edges.
        Rather than representing when an edge appears, the encoding captures which temporal edges can be taken consecutively along a temporal path.
        In contrast to static graphs -- where succession is determined simply by shared endpoints -- this relation additionally depends on the order at which the edges are present.
        As a result, the encoding directly represents temporal reachability via a partial order on temporal edges.
        \begin{definition}[degree encoding]
        \label{def:degree_encoding}
            Let $\gcal=(V,\mathcal{E})$ be an undirected strict \tg. The \emph{degree encoding} \relGraphDegree{\gcal} is the relational structure defined as:
            \vspace{0.4em}
            
            
            \noindent
            \begin{minipage}[t]{0.2\linewidth}
            \DescLabel{universe}
            \end{minipage}%
            \hspace{0.5em}
            \begin{minipage}[t]{0.8\linewidth}
            $U = V \cup \mathcal{E}$
            \end{minipage}

            \noindent
            \begin{minipage}[t]{0.2\linewidth}
            \DescLabel{unary relations}
            \end{minipage}%
            \hspace{0.5em}
            \begin{minipage}[t]{0.8\linewidth}
            $V(\cdot),\ecal(\cdot)$
            \end{minipage}

            \noindent
            \begin{minipage}[t]{0.2\linewidth}
            \DescLabel{binary relations}
            \end{minipage}%
            \hspace{0.5em}
            \begin{minipage}[t]{0.8\linewidth}
            \begin{itemize}[leftmargin=*, topsep=0pt, itemsep=0pt, parsep=0pt, partopsep=0pt]
                \item $\inc \subseteq \ecal \times V$, where
                $\inc((uv,t),x) \Leftrightarrow x=u \text{ or } x=v$,
                \item $\psuc \subseteq \ecal \times \ecal$, where
                $\psuc((e_1,t_1),(e_2,t_2)) \Leftrightarrow
                t_1 < t_2 \text{ and there exists } v \in V \text{ such that }
                \inc(e_1,v) \text{ and } \inc(e_2,v)$.
            \end{itemize}
            \end{minipage}
        \end{definition}
    \iflong
    \begin{figure}[h]
        \centering
        \includegraphics[width=0.75\linewidth]{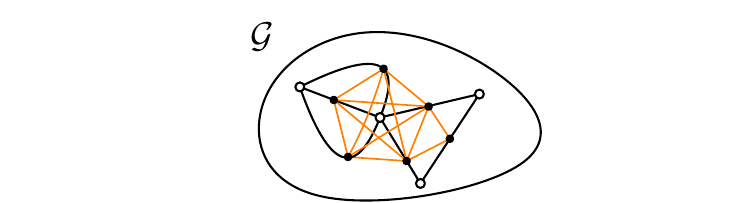}
        \caption{Illustration of the degree encoding $\relGraphDegree{\gcal}$ of a temporal graph $\gcal$.
        Each temporal edge $(e,t) \in \ecal$ between two vertices (empty circles) is represented by a temporal-edge node (filled circles); this is illustrated as multiple parallel edges, one for each temporal presence of a static edge.
        Two such temporal-edge-nodes are related in the possible-successor relation whenever the corresponding temporal edges share an endpoint and satisfy the temporal order.
        \iflong This may induce a clique among all temporal edges incident to the same vertex.\fi
        }
        \label{fig:gaifman-degree_encoding}
    \end{figure}
    \fi
    
    \iflong\paragraph*{\vimText encoding}\fi
    \label{subsec:vim_encoding}
        The \vimText encoding directly represents the \vimText decomposition of the temporal graph at the level of the relational structure.
        The encoding contains a sequence of \emph{bag elements}, one for each bag of the \vimText decomposition.
        Vertices and temporal edges are related to the bag elements in which they are \alive, and successive bags are connected by a linear successor relation.
        In contrast to the lifetime encoding, the bags are not \textit{pairwise} adjacent; this avoids introducing a clique over time and ensures that the density of the structure is bounded by the \vimText width of the graph.        
        \begin{definition}[\vimText encoding]
        \label{def:vim_encoding}
            Let $\gcal=(V,\mathcal{E})$ be a directed strict \tg with \vimText decomposition $(\vtimbag_t)_{t\in[\lifetime]}$. The \emph{relational \vimText structure} \relGraphVim{\gcal} is defined as:
            \vspace{0.4em}
            
            
            \noindent
            \begin{minipage}[t]{0.2\linewidth}
            \DescLabel{universe}
            \end{minipage}%
            \hspace{0.5em}
            \begin{minipage}[t]{0.8\linewidth}
            $U=V \cup \mathcal{E}\cup \vtimbagset$
            \end{minipage}

            \noindent
            \begin{minipage}[t]{0.2\linewidth}
            \DescLabel{unary relations}
            \end{minipage}%
            \hspace{0.5em}
            \begin{minipage}[t]{0.8\linewidth}
            $V(\cdot),\ecal(\cdot), \vtimbagset(\cdot)$
            \end{minipage}

            
            \noindent
            \begin{minipage}[t]{0.2\linewidth}
            \DescLabel{binary relations}
            \end{minipage}%
            \hspace{0.5em}
            \begin{minipage}[t]{0.8\linewidth}
            \begin{itemize}[leftmargin=*, topsep=0pt, itemsep=0pt, parsep=0pt, partopsep=0pt]
                \item $\inc\subseteq \mathcal{E}\times V$
                where $\inc((uv,t),x)\Leftrightarrow x=u \text{ or } x=v$,
                \item $\bag\subseteq V\times \vtimbagset$
                where $\bag(v,\vtimbag_t)\Leftrightarrow v\in \vtimbag_t$,
                \item $\pres\subseteq \mathcal{E}\times \vtimbagset$
                where $\pres((uv,t),\vtimbag_i)\Leftrightarrow
                \tau(i)=t \text{ and } u,v\in \vtimbag_i$,
                \item $\bagbefore\subseteq \vtimbagset\times \vtimbagset$
                where $\bagbefore(\vtimbag_i,\vtimbag_j)\Leftrightarrow
                \tau(i)+1=\tau(j)$.
            \end{itemize}
            \end{minipage}
        \end{definition}
        \ifshort
        \begin{figure}[t]
            \centering
            \iflong
            \includegraphics[width=0.85\linewidth]{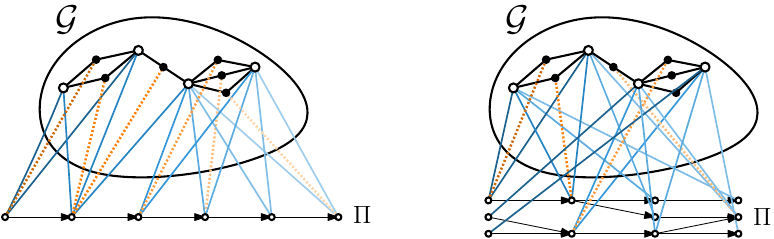}
            \fi
            \caption{
            Illustration of the \vimText encoding \relGraphVim{\gcal} (left) and the \timText encoding \relGraphTim{\gcal} (right).
            Each temporal edge $(e,t) \in \ecal$ between two vertices (empty circles) is represented by a temporal-edge node (filled circles).            
            Vertices are connected via the bag relation (\iflong$\bag$, shown in \fi{}solid blue) to every bag in which they are \alive.
            Temporal-edge-nodes are connected analogously via the presence relation (\iflong$\mathsf{pres}$, shown in \fi{}dotted orange).
            The relation~$\bagbefore$ induces a path (left), resp. tree (right), over the bag-nodes \vtimbagset.
            }
            \label{fig:gaifman-vim_encoding}\label{fig:gaifman-tim_encoding}
        \end{figure}
        \fi
    
    \iflong\paragraph*{\timText encoding}\fi
    \label{subsec:tim_encoding}
        \iflong
        Just as a \timText decomposition generalises the \vimText decomposition of a temporal graph by allowing multiple bags per time, the \timText encoding is a generalisation of the \vimText encoding. 
        While a \vimText decomposition induces a path over the bags, a \timText decomposition organizes the bags in a tree structure.  \fi
        In a \timText decomposition, vertices in the same connected component of a snapshot must be in the same bag at that time, while those in different connected components may be in different bags. 
        \iflong
        Note that, since times are not part of the universe of this encoding, we cannot determine the time associated to a given bag without a formula (of length possibly lifetime) which finds its distance from a leaf bag labelled with time 1. This forces us to focus on the \emph{relative} times of pairs of bags, as we do in the degree encoding.
        \fi
        \begin{definition}[\timText encoding]
        \label{def:tim_encoding}
            Let $\gcal=(V,\mathcal{E})$ be a directed strict \tg with \timText decomposition $(T,\vtimbagset,\tau)$. The \emph{relational \timText structure} \relGraphTim{\gcal} is defined as:
            \vspace{0.4em}
            
            \noindent
            \begin{minipage}[t]{0.2\linewidth}
            \DescLabel{universe}
            \end{minipage}%
            \hspace{0.5em}
            \begin{minipage}[t]{0.8\linewidth}
            $U=V \cup \mathcal{E}\cup \vtimbagset$
            \end{minipage}

            \noindent
            \begin{minipage}[t]{0.2\linewidth}
            \DescLabel{unary relations}
            \end{minipage}%
            \hspace{0.5em}
            \begin{minipage}[t]{0.8\linewidth}
            $V(\cdot),\ecal(\cdot), \vtimbagset(\cdot)$
            \end{minipage}
            \noindent
            \begin{minipage}[t]{0.2\linewidth}
            \DescLabel{binary relations}
            \end{minipage}%
            \hspace{0.5em}
            \begin{minipage}[t]{0.8\linewidth}
            \begin{itemize}[leftmargin=*, topsep=0pt, itemsep=0pt, parsep=0pt, partopsep=0pt]
                \item $\inc\subseteq \mathcal{E}\times V$
                where $\inc((uv,t),x)\Leftrightarrow x=u \text{ or } x=v$,
                \item $\bag\subseteq V\times \vtimbagset$
                where $\bag(v,\vtimbag_i)\Leftrightarrow v\in \vtimbag_i$,
                \item $\pres\subseteq \mathcal{E}\times \vtimbagset$
                where $\pres((uv,t),\vtimbag_i)\Leftrightarrow
                \tau(i)=t \text{ and } u,v\in\vtimbag_i$,
                \item $\bagbefore\subseteq \vtimbagset\times \vtimbagset$
                where $\bagbefore(\vtimbag_i,\vtimbag_j)\Leftrightarrow
                \tau(i)+1=\tau(j) \text{ and } (\vtimbag_i,\vtimbag_j)\in E(T)$.
            \end{itemize}
            \end{minipage}
        \end{definition}
        \iflong
        \begin{figure}[t]
            \centering
            \iflong
            \includegraphics[width=0.85\linewidth]{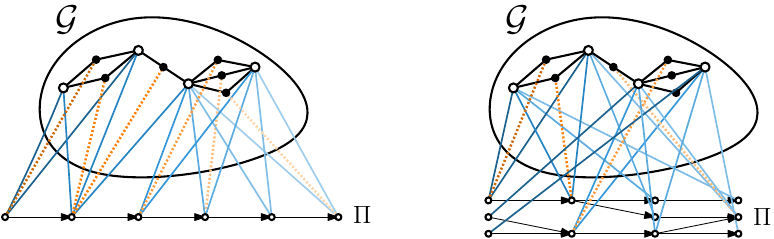}
            \else
            \includegraphics[width=0.75\linewidth]{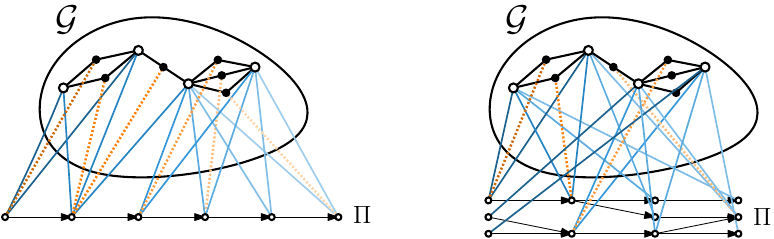}
            \fi
            \caption{
            Illustration of the \vimText encoding \relGraphVim{\gcal} (left) and the \timText encoding \relGraphTim{\gcal} (right).
            Each temporal edge $(e,t) \in \ecal$ between two vertices (empty circles) is represented by a temporal-edge node (filled circles).            
            Vertices are connected via the bag relation (\iflong$\bag$, shown in \fi{}solid blue) to every bag in which they are \alive.
            Temporal-edge-nodes are connected analogously via the presence relation (\iflong$\mathsf{pres}$, shown in \fi{}dotted orange).
            The relation~$\bagbefore$ induces a path (left), resp. tree (right), over the bag-nodes \vtimbagset.
            }
            \label{fig:gaifman-vim_encoding}\label{fig:gaifman-tim_encoding}
        \end{figure}
        \fi
\section{Temporal Meta--Theorems}
\label{sec:temporal-meta}
    \ifshort
    In this section, we establish meta-theorems for FO and MSO model checking on temporal graphs.
    Using the relational encodings from \Cref{sec:graph_encodings}, we show that FO/MSO-definable temporal properties are fixed-parameter tractable under bounded (i)~$\lifetime$, (ii)~$\tdegree$, (iii)~$\vim$, or (iv)~$\tim$, possibly in combination with structural assumptions on the footprint.

    \temporalMSOmeta*
    \temporalFOmeta*

    Both results are obtained by analysing the \textit{Gaifman graph} $\gaifman(\scal)$ of a relational structure $\scal\in\{\relGraphLifetime{\gcal}, \relGraphDegree{\gcal}, \relGraphVim{\gcal}, \relGraphTim{\gcal}\}$ with universe $U$ and a set of relations $\mathfrak{R}$. The Gaifman graph is a static, simple, and undirected graph with vertex set $U$ and an edge between two elements $a,b\in U$ whenever they are in a relation $R\in\mathfrak{R}$.
    \else
    The temporal relational encodings from \Cref{sec:graph_encodings} allow the application of the FO and MSO model checking meta-theorems.
    Concretely, we analyse the structure of these encodings via the \emph{Gaifman graph} of each relational structure.
    This yields temporal meta-theorems linking logical expressibility to algorithmic tractability: once a temporal property is expressed in FO/MSO over a suitable encoding, fixed-parameter tractability follows from the sparsity of the corresponding Gaifman graphs.
    \begin{definition}[Gaifman graph]
        Let $\scal$ be any finite relational structure with universe $U$ and some set of relations $\mathfrak{R} =\{R_1,R_2,\dots\}$.
        The \emph{Gaifman graph} $\gaifman(\scal)$ is a static, simple, undirected graph with vertex set $U$ and an edge $\{a,b\}$ if and only if (i) $a\neq b$ and (ii) there exist a relation $R\in\mathfrak{R}$ with $k\ge 2$ arguments and a tuple $(a_1,\dots,a_k)\in R$ such that $a,b\in\{a_1,\dots,a_k\}$.
    \end{definition}
    The illustrations in \Cref{fig:gaifman-degree_encoding,fig:gaifman-lifetime_encoding,fig:gaifman-vim_encoding,fig:gaifman-tim_encoding} depict the Gaifman graphs of our four encodings.

    In the next two subsections, we prove that for each of the parameters \lifetime, \tdegree, \vim, and \tim, the Gaifman graph of the corresponding encoding preserves bounded treewidth (for MSO) and nowhere denseness (for FO), up to a factor depending only on the parameter.
    \fi

    \subsection{Temporal Encodings preserve Treewidth}
        MSO model checking is fixed-parameter tractable on graph classes of bounded treewidth (\Cref{thm:MSO-MC}).
        To apply the MSO meta-theorem to temporal graphs, it suffices to show that the Gaifman graph of each encoding has bounded treewidth whenever the corresponding parameter \lifetime, \tdegree, \vim or \tim is bounded.
        \iflong
        This has previously been established for the lifetime encoding~\cite{haag_feedback_2022} and the degree encoding~\cite{enright_deleting_2021}; we provide proofs for completeness and unified notation. For the \vimText encoding and \timText encoding introduced in this work, we give novel proofs.
        
        \fi
        \ifshort
        This treewidth-preservation bound has been established for the lifetime and degree encodings by \cite{haag_feedback_2022,enright_deleting_2021}, while the results for the VIM and TIM encodings are novel.

        All results in this section are proven with a similar technique. Starting from a tree, \vimText, or \timText decomposition of the temporal graph \gcal, we locally extend the bags with temporal elements specific to each encoding to obtain a tree decomposition of the Gaifman graph. Crucially, the size increase for each bag is always in relation to the parameter corresponding to the encoding.
        \else        
        All proofs in this section follow the same structural pattern. We start from a (tree, \vimText, or \timText) decomposition and locally adjust its bags to obtain a tree decomposition of the Gaifman graph. The proofs differ only in the underlying decomposition and in the resulting blow-up of the bag sizes. Since the individual arguments are largely parallel, readers familiar with one proof may safely skim the others, focusing on the encoding-specific bag constructions and width bounds.
        \fi
        \begin{lemma}[\cite{haag_feedback_2022}]
        \label{lem:lifetime_enc-preserves-tw}
            Let \gcal be a \tg with treewidth $\tw(\gcal)$ and lifetime~\lifetime. 
            Then
            \ifshort
            $\tw\big(\gaifman(\relGraphLifetime{\gcal})\big)\ \le\ \Lambda+\Lambda\tw(\gcal)^2+\tw(\gcal)$ and, consequently, $\tw\big(\gaifman(\relGraphLifetime{\gcal})\big)\in\mathcal{O}(\tw(\gcal)^2\lifetime)$.
            \else
            \[
            \tw\big(\gaifman(\relGraphLifetime{\gcal})\big)\ \le\ \Lambda+\tw(\gcal)^2\quad \text{and, consequently,}\quad \tw\big(\gaifman(\relGraphLifetime{\gcal})\big)\in\mathcal{O}(\tw(\gcal)^2 + \lifetime).
            \]
            \fi
        \end{lemma}
        \iflong
        \begin{proof}
            This result can be found in \cite[Theorem 23]{haag_feedback_2022}. For completeness, we give a proof here.
            Let $G$ be the footprint of $\gcal$ and let $(T,B)$ be an optimal tree decomposition of $G$, \ie $\max_{i\in V(T)}|B_i|-1=\tw(\gcal)$.
            We build a tree decomposition $(T,\{B'_i:i\in V(T)\})$ of $\gaifman(\relGraphLifetime{\gcal})$ as follows. Note that both tree decompositions are indexed by the same tree $T$.
            
            For each $i\in V(T)$, define the bag
            \[
              B'_i \ :=\ B_i\ \cup\ \{(uv,t)\in \mathcal{E}: u\in B_i \wedge v\in B_i\wedge t\in L\}\ \cup\ L,\ 
            \]
            where $B_i$ is the bag of the tree decomposition of $G$ indexed by $i$.
            
            We show that this forms a tree decomposition. Recall that the universe of $\relGraphLifetime{\gcal}$ is $U = V \cup E \cup L$.
            \vspace{-0.5em}\begin{enumerate}[label=(\roman*)]
                \item First, we show $\bigcup_{i\in V(T)} B'_i = U$:
                \begin{itemize}
                    \item A vertex $v\in V$ lies in every $B'_i$ where $v\in B_i$ in $(T,B)$. Since $(T,B)$ is a tree decomposition, there is at least one such bag.
                    \item A temporal edge $\eps=(uv,t)\in \mathcal{E}$ lies in every $B'_i$ such that $u\in B_i$ and $v\in B_i$. Since $(T,B)$ is a tree decomposition of $G$, there exists a bag containing both endpoints of $e$.
                    \item A time $t\in L$ lies in every bag of the decomposition.
                \end{itemize}
                \item
                Every edge of the Gaifman graph (induced by a binary relation of the encoding) is covered by some bag:
                \vspace{-0.5em}\begin{itemize}
                    \item $\inc\subseteq \mathcal{E}\times V$: if $v$ is an endpoint of $\eps$ then there exists a bag containing both $v$ and $\eps$ by above reasoning. Therefore we have $\eps,v\in B'_i$ for some $i\in V(T)$.
                    \item $\pres\subseteq \mathcal{E}\times L$: since each time $t$ appears in every bag and each temporal edge appears in at least one bag, $\eps,t\in B'_i$ for some $i\in V(T)$.
                    \item $\timebefore\subseteq L \times L$: all times are in every bag. As a result $t_1,t_2\in B'_i$ for all $i\in V(T)$.
                \end{itemize}
                \item 
                For each element of $U$, the indices of bags containing it form a connected subtree of $T$:
                \vspace{-0.5em}\begin{itemize}
                  \item $v\in V$ lies in exactly the same bags as in $(T,B)$. This forms a connected subtree of~$T$.
                  \item $\eps\in \mathcal{E}$ lies in every bag containing both its endpoints. Since $(T,B)$ is a tree decomposition, the indices of bags containing each endpoint must form a connected subtree of~$T$. Thus, the intersection of these subtrees is also a connected subtree of $T$.
                  \item $t\in L$ lies in all bags, thus the subtree is $T$ itself.
                \end{itemize}
            \end{enumerate}
            Lastly, we show the treewidth bound. For $i\in V(T)$, we have $|B_i'|\le \Lambda + \Lambda|B_i|^2 + |B_i|$, and hence $\tw(\gaifman(\relGraphLifetime{\gcal}))\in O(\tw(\gcal)^2\lifetime)$.
        \end{proof}
        \fi

        \begin{lemma}[\cite{enright_deleting_2021}]
        \label{lem:tdegree_enc-preserves-tw}
            Let $\gcal$ be a temporal graph with treewidth $\tw(\gcal)$ and maximum temporal degree $\tdegree$.
            Then \ifshort $\tw\big(\gaifman(\relGraphDegree{\gcal})\big)\ \le\ (\tdegree+1)(\tw(\gcal)+1)$ and, consequently, $\tw\big(\gaifman(\relGraphDegree{\gcal})\big)\in\mathcal{O}(\tw(\gcal)\cdot \tdegree)$.
            \else
            \[\tw\big(\gaifman(\relGraphDegree{\gcal})\big)
            \ \le\ (\tdegree+1)(\tw(\gcal)+1)
            \quad \text{and, consequently,}\quad
            \tw\big(\gaifman(\relGraphDegree{\gcal})\big)\in\mathcal{O}(\tdegree\cdot\tw(\gcal)).
            \]
            \fi
        \end{lemma}
        \iflong
        \begin{proof}
            This result can be found in \cite[Lemma 5.3]{enright_deleting_2021}. For completeness, we give a proof here.
            Let $G$ be the footprint of $\gcal$ and let $(T,B)$ be an optimal tree decomposition of $G$, \ie $\max_{i\in V(T)}|B_i|-1=\tw(\gcal)$.
            We build a tree decomposition $(T,\{B'_i:i\in V(T)\})$ of $\gaifman(\relGraphDegree{\gcal})$ as follows. Note that both tree decompositions are indexed by the same tree $T$.
            
            For each $i\in V(T)$, define the bag
            \[
              B'_i \ :=\ B_i\ \cup\ \{(e=uv,t) \in \mathcal{E}: u\in B_i \vee v\in B_i\},\ 
            \]
            where $B_i$ is the bag of the tree decomposition of $G$ indexed by $i$. That is, $B'_i$ contains all vertices of $B_i$ and all temporal edges incident to them.
            
            We show that this forms a tree decomposition. Recall that the universe of $\relGraphDegree{\gcal}$ is $U = V \cup \ecal$.
            \begin{enumerate}[label=(\roman*)]
                \item First, we show $\bigcup_{i\in V(T)} B'_i = U$:
                \begin{itemize}
                    \item  A vertex $v\in V$ lies in every $B'_i$ where $v\in B_i$ in $(T,B)$. Since $(T,B)$ is a tree decomposition, there is at least one such bag.
                    \item A temporal edge $(e=uv,t)\in \ecal$ lies in every $B'_i$ such that $u\in B_i$ or $v\in B_i$. Since $(T,B)$ is a tree decomposition, there is at least one such bag.
                \end{itemize}
                \item 
                Every edge of the Gaifman graph (induced by a binary relation of the encoding) is covered by some bag:
                \vspace{-0.5em}\begin{itemize}
                    \item $\inc\subseteq \mathcal{E}\times V$: if $v$ is an endpoint of $(e,t)$ then there exists a bag containing both $v$ and $(e,t)$ by above reasoning. Therefore we have $(e,t),v\in B'_i$ for some $i\in V(T)$.
                    \item $\psuc\subseteq \mathcal{E}\times \mathcal{E}$: if $(e_1,t_1)$ is a possible successor of $(e_2,t_2)$, there must be a vertex $v$ which is a shared endpoint of the edges. Since both temporal edges are contained in all bags containing $v$, they must be in a bag together.
                \end{itemize}
                \item 
                For each element of $U$, the indices of bags containing it form a connected subtree of $T$:
                \vspace{-0.5em}\begin{itemize} 
                    \item $v\in V$ lies in exactly the same bags as in $(T,B)$. This forms a connected subtree of $T$.
                    \item $(e,t)\in\mathcal{E}$ lies in every bag containing one or both of its endpoints. 
                    Since $(T,B)$ is a tree decomposition, the indices of bags containing an endpoint form a connected subtree of $T$. Thus, the bags containing at least one endpoint form the union of two connected subtrees with nonempty intersection, and hence a connected subtree.
                \end{itemize}
            \end{enumerate}
            Lastly, we show the treewidth bound.
            For $i\in V(T)$, we have $|B_i'|\le |B_i| + |B_i|\tdegree = (\tdegree+1)|B_i|$, hence $\tw(\gaifman(\relGraphDegree{\gcal}))\in\mathcal{O}(\tw(\gcal)\cdot\tdegree)$.
            \end{proof}
        \fi

        \newcommand{\vimbagset}{\ensuremath{{\bf\Gamma}}\xspace}
        \newcommand{\vimbag}{\ensuremath{\Gamma}\xspace}
        \iflong
        We now turn to the bounds for \vimText and \timText. First, we show the statement for the \vimText encoding. Recall that $\vim(\gcal) \ge \pw(\gcal) \ge \tw(\gcal)$.\fi
        \begin{lemma}
        \label{lem:vim_enc-preserves-tw}
            Let $\gcal$ be a temporal graph with \vimText width $\vim(\gcal)$.
            Then 
            \ifshort $\tw\big(\gaifman(\relGraphVim{\gcal})\big)\ \le\ \vim(\gcal)+2\vim(\gcal)^2$ and, consequently, $\tw\big(\gaifman(\relGraphVim{\gcal})\big)\in\mathcal{O}(\vim(\gcal)^2)$.
            \else
            \[
            \tw\big(\gaifman(\relGraphVim{\gcal})\big)
            \ \le\ \vim(\gcal)+2\vim(\gcal)^2
            \quad \text{and, in particular,}\quad
            \tw\big(\gaifman(\relGraphVim{\gcal})\big)\in\mathcal{O}(\vim(\gcal)^2).
            \]
            \fi
        \end{lemma} 
        \iflong
            \begin{proof}
            Let $\vimbagset=(\vimbag_t)_{t\in[\Lambda]}$ be the \vimText decomposition of \gcal, so $\vimbag_t=\{v\in V:t\in A(v)\}$ and
            $\max_t|\vimbag_t|=\vim$. Recall that the \vimText decomposition forms a (not necessarily optimal) path decomposition of the footprint~$G$ of \gcal.
            We build a \emph{path} decomposition $(T,\{B_t:t\in[\Lambda]\})$ of $\gaifman(\relGraphVim{\gcal})$ as follows. 
            To avoid confusion, we will refer to the bags $\vimbag_t$ of the \vimText decompositions of \gcal as \textit{\vimText-bags} and the bags $B_t$ of the path decomposition of $\gaifman(\relGraphVim{\gcal})$ as simply \textit{bags}. 
            
            For notational convenience, let $\vimbag_{\Lambda+1} := \emptyset$. For each $t\in V(T) = [\Lambda]$, define the bag
            \begin{align*}
                 B_{t} \ :=\ &  
                    \underbrace{\{v:v\in \vimbag_t\}}_{\text{$t$-\alive vertices}} \ \cup\ 
                    \underbrace{\{\eps\in \mathcal{E}: \eps = (e,t)\}}_{\text{temporal edges in $G_t$}}\ \cup \ 
                    \underbrace{\{\vimbag_t,\vimbag_{t+1}\}}_{\text{current and next VIM-bag}}.
            \end{align*}
            We show that this forms a path decomposition. Recall that the universe of \relGraphVim{\gcal} is $U=V\cup \mathcal{E}\cup \vimbagset$.
            \vspace{-0.5em}\begin{enumerate}[label=(\roman*)]
                \item First, we show $\bigcup_{t\in [\lifetime]} B_t = U$:
                \begin{itemize}
                    \item A vertex $v\in V$ lies in every $B_t$ with $t\in A(v)$.
                    \item A temporal edge $\eps\in \mathcal{E}$ lies in $B_t$ where $\eps= (e,t)$ for some static edge $e$, hence in some bag.
                    \item A \vimText-bag $\vimbag_t\in \vimbagset$ lies in the bags $B_t$ and $B_{t-1}$.
                \end{itemize}
                \item
                Every edge of the Gaifman graph (induced by a binary relation of the encoding) is covered by some bag:
                \vspace{-0.5em}\begin{itemize}
                    \item $\inc\subseteq \mathcal{E}\times V$: if $v$ is an endpoint of $\eps=(e,t)$
                    then $t\in A(v)$ and we must have $\eps,v \in B_t$.
                    \item $\Sedgetime\subseteq \mathcal{E}\times \vimbagset$: if $\Sedgetime(\eps,\vimbag_t)$, then $\eps,\vimbag_t\in B_t$.
                    \item $\bag\subseteq V\times \vimbagset$: if $\bag(v,\vimbag_t)$, then $v,\vimbag_t \in B_t$.
                  \item $\bagbefore\subseteq \vimbagset \times \vimbagset$: for any $t\in[\Lambda)$, we have $\vimbag_{t},\vimbag_{t+1}\in B_t$ and $B_{t+1}$.
                \end{itemize}
                \item 
                For each element of $U$, the indices of bags containing it form a connected subtree of $T$:
                \vspace{-0.5em}\begin{itemize}
                  \item $v\in V$ occurs precisely in the bags $B_t$ where $t\in A(v)$. Since $A(v)$ is a contiguous interval, this must be a subpath of the decomposition.
                  \item $\eps = (e,t)\in \mathcal{E}$ occurs only in the bag $B_t$. This must trivially induce a connected subpath of the decomposition.
                  \item $\vimbag_t\in F$ occurs only in $B_t$ and $B_{t-1}$, which are adjacent in the decomposition.
                \end{itemize}
            \end{enumerate}
            Lastly, we bound the width.
            For any $t\in[\Lambda]$, we have $|\vimbag_t|\le\vim$.
            Every temporal edge in $B_t$ has both endpoints in $\vimbag_t$, so there are
            at most $|\vimbag_t|^2\le \vim^2$ such edges.
            Thus $|B_t| \le 1 + |\vimbag_t| + |\vimbag_t|^2 \le 1 + \vim + \vim^2$, and hence $\tw(\gaifman(\relGraphVim{\gcal})) \in \mathcal{O}(\vim(\gcal)^2)$.
        \end{proof}
        \fi

        \iflong
        Second, we show the statement for the \timText encoding. Recall that $\tim(\gcal) \ge \tw(\gcal)$.\fi
        \begin{lemma}
        \label{lem:tim_enc-preserves-tw}
            Let $\gcal$ be a temporal graph with \timText width $\tim(\gcal)$.
            Then 
            \ifshort
            $\tw\big(\gaifman(\relGraphTim{\gcal})\big)\ \le\ \tim(\gcal)^2+3\tim(\gcal) -1$ and, consequently, $\tw\big(\gaifman(\relGraphTim{\gcal})\big)\in\mathcal{O}(\tim(\gcal)^2)$.
            \else
            \[
            \tw\big(\gaifman(\relGraphTim{\gcal})\big)
            \ \le\ \tim(\gcal)^2+3\tim(\gcal) -1
            \quad \text{and, in particular,}\quad
            \tw\big(\gaifman(\relGraphTim{\gcal})\big)\in\mathcal{O}(\tim(\gcal)^2).
            \]
            \fi
        \end{lemma}
        \iflong
        \newcommand{\timbagel}{\ensuremath{\Pi}\xspace}
        \begin{proof}
            Let $(T,\vtimbagset=\{\vtimbag_i\colon i\in V(T)\},\tau)$ be a \timText decomposition of \gcal.
            We build a tree decomposition $(T,\{B_i:i\in V(T)\})$ of $\gaifman(\relGraphTim{\gcal})$ as follows.
            To avoid confusion, we will refer to the bags $\vtimbag_i$ of the \timText decomposition as \textit{\timbagTEXTs} and the bags $B_i$ of the tree decomposition as simply \textit{bags}.
            
            For each $i\in V(T)$, define the bag
            \begin{align*}
                B_i  := &
                \underbrace{\{v\in \vtimbag_i\}}_{\text{vertices in \timbagTEXT}}
                \cup\
                \underbrace{\{\eps=(e,t)\in \mathcal{E}: \tau(i)=t,\ \eps[0],\eps[1]\in \vtimbag_i\}}_{\text{temporal edges with endpoints in \timbagTEXT}}
                \ \cup
                \underbrace{\{\timbagel_i, p(\timbagel_i)\}}_{\text{\timbagTEXT and parent}} ,
            \end{align*}
            where $p(\vtimbag_i)$ is the parent \timbagTEXT of $\vtimbag_i$ in the \timText decomposition.
            
            We show that this forms a tree decomposition. Recall that the universe of \relGraphTim{\gcal} is $U=V\cup E\cup \vtimbagset$.
            \vspace{-0.4em}\begin{enumerate}[label=(\roman*)]
                \item First, we show $\bigcup_{i\in V(T)} B_i = U$:
                    \begin{itemize}
                        \item A vertex $v\in V$ lies in every $B_i$ with $v\in \vtimbag_i$. By definition of a \timText decomposition there exists at least one such \timbagTEXT for every vertex.
                        \item A temporal edge $\eps=(e,t)\in \mathcal{E}$ lies in $B_i$ where $\tau(i)=t$ and $\eps[0],\eps[1]\in \vtimbag_i$. By definition of a \timText decomposition, there exists at least one such \timbagTEXT for every temporal edge.
                        \item A \timbagTEXT $\vtimbag_i\in \vtimbagset$ lies in $B_i$. 
                    \end{itemize}
                \item  
                Every edge of the Gaifman graph (induced by a binary relation of the encoding) is covered by some bag:
                    \begin{itemize}
                        \item $\inc\subseteq \mathcal{E}\times V$: if $v$ is an endpoint of $\eps=(e,t)$, then there exists $i\in V(T)$ with $\tau(i)=t$ and $\eps[0],\eps[1]\in \vtimbag_i$; hence $\eps,v\in B_i$.
                      \item $\mathsf{bag}\subseteq V\times \vtimbagset$: for $v\in \vtimbag_i$ we have $v,\vtimbag_i\in B_i$.
                      \item $\Sedgetime\subseteq \mathcal{E}\times \vtimbagset$: if $\Sedgetime(\eps,\vtimbag_i)$, then $\eps,\vtimbag_i\in B_i$.
                      \item $\bagbefore\subseteq \vtimbagset\times \vtimbagset$: for $(i,j)\in E(T)$, we have that $\vtimbag_j$ is the parent of $\vtimbag_i$ in $T$ by definition of the \timText encoding. By construction of our bags, we have $\timbagel_i\in B_i$ and $p(\timbagel_i)=\timbagel_j\in B_i$. Hence $\timbagel_i$ and $\timbagel_j$ occur together in the bag $B_i$.
                    \end{itemize}
                \item 
                For each element of $U$, the indices of bags containing it form a connected subtree of $T$:
                    \begin{itemize}
                      \item $v\in V$ appears in the bag $B_i$ if and only if $v\in \vtimbag_i$. By the definition of a TIM decomposition, the subtree of $T$ induced by \timbagTEXTs containing $v$ must form a directed path. Since the \timText decomposition and this tree decomposition are indexed by the same tree, the bags containing $v$ form a subpath of $T$.
                      \item $\eps=(e,t)\in \mathcal{E}$ appears in the bag $B_i$ if and only if both endpoints of $\eps$ are in $\vtimbag_i$ and $\tau(i)=t$. This occurs in exactly one bag of the decomposition. Therefore, the subtree induced by the bags of $(T,\mathcal{B})$ containing $\eps$ is connected and non-empty.
                      \item $\vtimbag_i\in\vtimbagset$ appears in $B_i$ and all bags of its children, which induce a connected subtree of $T$. Since the tree decomposition and the TIM decomposition are indexed by the same tree, the claim follows.
                    \end{itemize}
                \end{enumerate}
            Lastly, we show the width bound. Since $|\vtimbag_i|\leq \tim$ for $i\in V(T)$, we have $|B_i|\ \le |\vtimbag_i|\ +\ |\mathcal{E}_i|\ +\ |N_T(\vtimbag_i)|\ \le\ 3\tim+\tim^2$, hence $\tw(\gaifman(\relGraphTim{\gcal}))\in\mathcal{O}(\tim(\gcal)^2)$.
        \end{proof}
        \fi
    
    \subsection{Temporal Encodings preserve Nowhere Denseness}
        FO model checking is fixed-parameter tractable on nowhere dense graph classes (\Cref{thm:FO-MC}), a broad family that includes, for example, planar graphs, graphs of bounded expansion, and graphs of bounded treewidth.
        To lift FO meta-theorems to temporal graphs, it therefore suffices to show that the Gaifman graphs arising from our relational encodings form nowhere dense classes.\ifshort
        ~Among the four encodings, the proof for the lifetime encoding requires the most care; we briefly sketch the main argument below.
        \else
        
        To the best our knowledge, FO meta-theorems have previously been considered for temporal graphs only in the work by Mans and Mathieson~\cite{mans_treewidth_2014} who studied a variation of the lifetime encoding. 
        At the time, tractability of FO model checking was known only for the more restricted class of bounded-expansion graphs, which forms a proper subclass of nowhere dense graph classes.
        While their approach aimed at establishing preservation of sparsity, a flaw in the encoding and the associated argument was later identified by Fomin et al.~\cite{fomin_as_2020}:
        Mans and Mathieson consider a generalisation of temporal graphs where the vertex set may vary over time. In their encoding of these \textit{dynamic} graphs, the universe contains an element for every vertex appearance (a temporal vertex), rather than one element per static vertex as in our encoding.
        They claimed one could identify whether two temporal vertices belong to the same underlying vertex using a constant-length formula. Fomin et al.~\cite{fomin_as_2020} later showed that this would imply the existence of a constant-length formula for 3-colourability on constant-treewidth graphs (NP-hard by Dailey~\cite{dailey_uniqueness_1980}), which would contradict the assumption that $P\neq NP$.

        In what follows, we show that all four encodings preserve nowhere denseness when their respective parameter is bounded.
        Among the four encodings, the lifetime encoding requires the most care. For the proof we require additional definitions and notation for depth-$r$ topological minors.

        \smallparagraph{Depth-$r$ topological minors via models.}
        Recall that for static graphs $G$ and $H$, we write $H\preceq_r G$, if some $r$-subdivision of $H$ is isomorphic to a subgraph of $G$. An $r$-subdivision is obtained by replacing a set of edges in $H$ by paths of length at most $r+1$ such that the \textit{internal vertices} (all but the endpoints) of all paths are disjoint.
        Equivalently, $H \preceq_r G$ holds if $G$ contains a \emph{depth-$r$ model} of $H$, which is a pair $\eta=(\eta_V,\eta_E)$ where
        \begin{enumerate}[leftmargin=*, itemsep=0pt, topsep=2pt]
            \item $\eta_V\colon V(H)\to V(G)$ is an injective mapping, and
            \item for every edge $ab\in E(H)$, $\eta_E(ab)$ is a simple path in $G$, called \textit{model-path}, of length at most $r+1$ which connects $\eta_V(a)$ and $\eta_V(b)$,
        \end{enumerate}
        such that for all distinct edges $ab,cd\in E(H)$, the paths $\eta_E(ab)$ and $\eta_E(cd)$ are internally vertex-disjoint, and no internal vertex of any path $\eta_E(ab)$ lies in $\eta_V(V(H))$.
        The vertices in $\eta_V(V(H))$ are called the \emph{branch-vertices} of the model.
        Intuitively, the model represents how the graph $H$ can be ``found'' inside $G$ after stretching each edge of $H$ into a short path of length at most $r+1$.
        The process of \textit{contracting} an edge $uv$ yields the graph $G/uv$ obtained by identifying $u$ and $v$ into a new vertex that is adjacent to every vertex that was adjacent to $u$ or $v$ in $G$, and then deleting self-loops and parallel edges (so that the result remains simple).
        \fi
        \begin{lemma}
        \label{lem:lifetime_enc-preserves-nwdense}
            Let $\mathcal{C}$ be a nowhere dense class of static graphs with denseness-function \densef, and consider the class of all temporal graphs $\mathcal{C'}$ with footprint in $\mathcal{C}$ and maximum lifetime \lifetime.
            Then the class of Gaifman graphs $\gaifman(\relGraphLifetime{\gcal})$ for each $\gcal\in \mathcal{C'}$ is nowhere dense. In particular,
            the function
            \iflong 
            \[ g(r) := \densef(\lceil r/2\rceil) + 2\lifetime \]
            \fi
            is a denseness-function for this class.
        \end{lemma}
        \newcommand{\gaifgraph}{\ensuremath{G_{\gaifman}}\xspace}
        \ifshort
        \begin{proof}[Proof sketch]
            The universe of the lifetime encoding contains static vertices, temporal edges, and time steps. Its Gaifman graph $\gaifgraph$ consists of a clique on the time steps $L$ and a bipartite incidence graph on $V\cup\ecal$, where every temporal edge $\eps=(uv,t)\in\ecal$ has degree~2 and is adjacent to $u$ and $v$. Hence, $\gaifgraph-L$ is a subdivided multi-edge version of the footprint.

            Since the footprint is nowhere dense, we can show that $\gaifgraph-L$ is nowhere dense too.
            We do this by showing any depth-$r$ clique topological minor in $\gaifgraph-L$ can be projected to a depth-$\lceil r/2\rceil$ clique topological minor in the footprint. Intuitively, $\gaifgraph-L$ already contains one subdivision vertex per footprint edge occurrence, so any subdivided edge of a clique minor alternates between $V$ and $\ecal$. Consequently, when projecting the model to the footprint, path lengths drop by a factor of~2 (up to ceilings).

            Adding $L$ can contribute at most $\lifetime$ vertices to the clique minor and at most $\lifetime$ subdivision vertices for clique edges. Thus the excluded clique size increases only by an additive $2\lifetime$.
        \end{proof}
        \fi
        \iflong
        \begin{proof}
            Let $\mathcal{C}$ be a nowhere dense static graph class with denseness-function \densef, and let $\mathcal{C'}$ be the class of all temporal graphs $\gcal=(G,\lambda)$ with lifetime at most~\lifetime and whose footprint $G$ lies in $\mathcal{C}$.
            We analyze the class of Gaifman graphs $\{\gaifman( \relGraphLifetime{\gcal} ) \colon \gcal\in \mathcal{C'}\}$.

            Fix $\gcal\in\mathcal{C'}$ and let $\gaifgraph := \gaifman(\relGraphLifetime{\gcal})$
            Recall that the vertices of $\gaifgraph$ are $V \cup \ecal \cup L$ and:
            \begin{itemize}
                \item edges between $V$ and $\ecal$ encode the incidence relation via $\inc$,
                \item edges inside $L$ form a clique (from the order $\timebefore$), and
                \item edges between $\ecal$ and $L$ encode the presence relation via $\pres$.
            \end{itemize}
            For an illustration of $\gaifgraph$ refer back to \Cref{fig:gaifman-lifetime_encoding}. We will refer to a vertex $a\in V \cup \ecal \cup L$ of \gaifgraph as \textit{vertex node}, \textit{temporal-edge node}, or \textit{time node}, respectively.

            \bigparagraph{Step 1: $\gaifgraph-L$ is nowhere dense with $\densef_{\gaifman-L}(r):=\densef(\lceil r/2\rceil)$.}\quad
            Consider $\gaifgraph - L$.
            This graph can be obtained from the footprint $G$ by subdividing every edge $uv\in E(G)$ once (resulting in a temporal-edge node $\eps$) and multiplying this $u-\eps-v$ path for every temporal edge $(uv,t)\in\ecal$. Conversely, $G$ can be obtained from $\gaifgraph-L$ by contracting every temporal-edge node with one of its neighbours.
            
            We show that if we can find an $r$-subdivision of a clique $K_q$ in $\gaifgraph-L$, then there is an $\lceil r/2\rceil$-subdivision of $K_q$ in $G$. 
            From this we conclude that $\densef_{\gaifman-L}(r):=\densef(\lceil r/2\rceil)$ is a denseness function of $\gaifgraph-L$.
              
            Fix $r\in\mathbb{N}$ and $q\ge 4$, and assume that $K_q \preceq_r (\gaifgraph-L)$. Let $\eta$ be a depth-$r$ model of $K_q$ in $\gaifgraph-L$. 
            In $\gaifgraph-L$, every temporal-edge node $\eps=(uv,t)\in\ecal$ has degree~2 and is adjacent exactly to the nodes $u$ and $v$. Consequently, no temporal-edge node in $\ecal$ can be a branch-vertex of a depth-$r$ clique model when $q\ge 4$, and hence all branch-vertices must lie in $V$.
            Additionally, each model-path of length $\ell$ in $\gaifgraph-L$ must alternate between $V$ and $\ecal$, and therefore becomes a path of length at most $\lceil \ell/2\rceil$ on $V$ in $G$ (after contracting all temporal-edge nodes with one of its neighbours).
            Internal vertex-disjointness of the model-paths is preserved under contraction.
            Thus, we obtain a depth-$\lceil r/2\rceil$ model of $K_q$ in $G$, and therefore $K_q \preceq_{\lceil r/2\rceil} G$.
            
            By definition of $\densef$ as a denseness-function for $\mathcal{C}$, we have $K_{\densef(\lceil r/2\rceil)} \not\preceq_{\lceil r/2\rceil} G$ for all $r$. By our arguments above, this implies $K_{\densef_{\gaifman-L} (r)} = K_{\densef(\lceil r/2\rceil)} \not\preceq_r (\gaifgraph-L)$ for all $r\in\mathbb{N}$, and hence
            $\densef_{\gaifman-L}(r)=\densef(\lceil r/2\rceil)$ is a denseness-function for the class of all such graphs $\gaifgraph-L$.
        
        \bigparagraph{Step 2: $\gaifgraph$ is nowhere dense with $\densef_{\gaifman}(r):=\densef_{\gaifman-L}(r)+2\lifetime$.}\quad
            Towards contradiction assume $K_\ell\preceq_r \gaifgraph$
            for $\ell := \densef_{\gaifman-L}(r) +2\lifetime$
            , and let $\eta$ be a depth-$r$ model of $K_\ell$ in \gaifgraph.
            
            Step~1 shows that $\gaifgraph-L$ cannot contain a depth-$r$ topological minor isomorphic to $K_{\densef_{\gaifman-L}(r)}$.
            Thus, the only way the model $\eta$ can ``beat'' this bound is by exploiting clique of time nodes $L$. Importantly, $L$ can help in two ways:
            (i) by hosting branch-vertices, and
            (ii) by serving as internal vertices on model-paths for clique edges. We show that both effects are limited by $|L|=\lifetime$.
            
            First, we consider the time nodes in $L$ that are branch-vertices. 
            Let $B$ be the set of branch-vertices of $\eta$ and denote by $B_L := B\cap L$ the branch-vertices in~$L$ and by $B_{\neg L}:=B\setminus L$ the branch-vertices in $\gaifgraph-L$.
            Since $|L|=\lifetime$, we have $|B_L|\le \lifetime$ and therefore $\lvert B_{\neg L}\rvert \geq \ell-\lifetime = \densef_{\gaifman-L}(r)+\lifetime$. 

            Next, we consider the time nodes in $L$ that are used by model-paths.
            Let $F$ be the set of clique edges $xy$ whose model path contains a time node of $L$ as an internal vertex. Since model-paths have to be internally disjoint, each time node in $L$ can be used by the model-path of at most one clique edge, and therefore $\lvert F\rvert\leq \lvert L\setminus B_L\rvert\le \lifetime$.
            Let $S$ be the set of endpoints (in $B_{\neg L}$) of the edges in $F$. Then $\lvert S\rvert\leq 2\lvert F\rvert\leq 2\lifetime$.
            
            Define $B' := B_{\neg L}\setminus S$ as the set of branch vertices outside $L$ whose model-paths do not use~$L$.
            It holds $\lvert B'\rvert \geq \lvert B_{\neg L}\rvert - 2\lifetime \geq \densef_{\gaifman-L}(r)$.
            We claim that restricting $\eta$ to the branch-vertices $B'$ yields a depth-$r$ model of $K_{\lvert B'\rvert}$ in $\gaifgraph-L$.
            To that end, consider any two distinct $x,y\in B'$.
            Then $x,y\notin L$ and $xy\notin F$ (otherwise $x\in S$ or $y\in S$). Hence the model-path for $xy$ does not go through $L$ and must lie entirely in $\gaifgraph-L$.
            Thus, $\gaifgraph-L$ contains $K_{|B'|}$ as a depth-$r$ topological minor, and in particular contains $K_{\densef_{\gaifman-L}(r)}$ as a depth-$r$ topological minor. This is a contradiction to $\densef_{\gaifman-L}$ being a denseness-function for $\gaifgraph-L$.

            Consequently, $K_{\densef_{\gaifman-L}(r) + 2\lifetime} \not \preceq_r \gaifgraph$ for all $r\in\mathbb{N}$, and hence $g(r)=\densef_{\gaifman-L}(r)+2\lifetime$ is a denseness-function for the class of all such graphs $\gaifgraph$.
        
        \bigparagraph{Step 3: Conclusion.} \quad
        In summary, for every $r \in \mathbb{N}$ and every graph $\gaifgraph = \gaifman( \relGraphLifetime{\gcal} )$ arising from a temporal graph $\gcal\in\mathcal{C}'$, we have $K_{\densef(\lceil r/2 \rceil) + 2\lifetime} \not\preceq_r \gaifgraph$.
        Equivalently, the function $g(r) := \densef(\lceil r/2\rceil) + 2\lifetime$ is a denseness-function for the class $\{\gaifman(\relGraphLifetime{\gcal}) : \gcal\in \mathcal{C'}\}$.
        \end{proof}
        \fi
        
        For the degree encoding, we show a stronger preservation result, namely that the maximum static degree of the Gaifman graph is bounded by the temporal degree of \gcal. Since graph classes of bounded maximum degree are nowhere dense, this immediately implies the desired result.
        \begin{lemma}
            \label{lem:tdegree_enc-preserves-nwdense}
            Let $\mathcal{C}$ be a nowhere dense class of static graphs with denseness-function \densef, and consider the class of all temporal graphs $\mathcal{C'}$ with footprint in $\mathcal{C}$ and maximum temporal degree~\tdegree.
            For $\gcal\in\mathcal{C'}$, let $H := \operatorname{Gf}(\relGraphDegree{\gcal})$ be the Gaifman graph of the degree encoding of \gcal. Then $\Delta(H) \le 2\tdegree$ and, consequently, the class of all such graphs $H$ is nowhere dense.
        \end{lemma}
        \iflong
        \begin{proof}
            Let $\gaifgraph:=\gaifman(\relGraphDegree{\gcal})$ be the Gaifman graph of the degree encoding $\relGraphDegree{\gcal}$.
            Recall from the degree encoding that the vertices of $\gaifgraph$ are $V \cup \mathcal{E}$, where
            $\mathcal{E}$ is the set of temporal edges.
            The binary relations are:
            \vspace{-0.4em}
            \begin{itemize}
                \item $\Sedgesource((e,t),v)$ and $\Sedgetarget((e,t),v)$, connecting every temporal edge to its endpoints; and
                \item $(e_1,t_1)\psuc(e_2,t_2)$, which holds iff the two temporal edges share an endpoint and $t_1 \prec t_2$.
            \end{itemize}
            \vspace{-0.4em}
            Thus, in the Gaifman graph $\gaifgraph$:
            \vspace{-0.4em}
            \begin{itemize}
                \item each temporal edge $(e,t)$ is adjacent to its two endpoints; and
                \item two temporal edges $(e_1,t_1)$ and $(e_2,t_2)$ are adjacent if and only if they share a vertex in $V$ and their times are comparable (which is always the case, since $\prec$ is a total order).
            \end{itemize}
            \vspace{-0.4em}
            We can now bound the maximum degree of $\gaifgraph$ in terms of the maximum temporal degree $\tdegree$.
            
            Let $v\in V$. By definition of $\tdegree$, there are at most $\tdegree$ temporal edges incident with $v$, so in $\gaifgraph$ the vertex $v$ has at most $\tdegree$ neighbours (all of them in $\mathcal{E}$). There are no edges in $\gaifgraph$ between two vertices of $V$.
            Hence $\deg_H(v) \le \tdegree$.
            
            Let $\eps = (e,t)\in\mathcal{E}$ where $e$ has endpoints $u,v\in V$ (for directed
            graphs, $u$ is the source and $v$ the target; the argument is identical).
            Then the neighbours of $\eps$ in $\gaifgraph$ are 
            (1) its endpoints $u$ and $v$ (two vertices in $V$); and 
            (2) all other temporal edges incident with $u$ or with $v$ (the $\psuc$-neighbours).
            By the definition of $\tdegree$, there are at most $\tdegree-1$ temporal edges at $u$ different from $(e,t)$, and at most $\tdegree-1$ such edges at $v$. 
            Therefore $\eps$ has at most $(\tdegree-1)+(\tdegree-1)$ neighbours in $\mathcal{E}$, plus its two endpoint vertices.
            In particular, $\deg_H(x) \le 2 + 2(\tdegree-1) \le 2\tdegree$.
            
            Combining both cases, the maximum degree of $\gaifgraph$ is bounded by $\Delta(\gaifgraph) \le 2\tdegree$.
        \end{proof}
        \fi
        
        For the \vimText encoding and the \timText encoding the nowhere dense preservation again follows from a stronger result: the parameters \vim and \tim upper bound the treewidth of the footprint of a temporal graph. Since graph classes with bounded treewidth are nowhere dense, the desired result follows directly from the treewidth bounds established in \Cref{lem:vim_enc-preserves-tw,lem:tim_enc-preserves-tw}.
        \begin{corollary}
        \label{lem:vim_enc-preserves-nwdense}
        \label{lem:tim_enc-preserves-nwdense}
            Let $\mathbf{C}$ be a class of temporal graph with \vimText width at most $\vim$. Then the class of all Gaifman graphs $\{\gaifman(\relGraphVim{\gcal})\,\colon\,\gcal\in\mathbf{C}\}$ is nowhere dense.
            Equivalently, let $\mathbf{C'}$ be a class of temporal graph with \timText width at most $\tim$. Then the class of all Gaifman graphs $\{\gaifman(\relGraphTim{\gcal})\,\colon\,\gcal\in\mathbf{C'}\}$ is nowhere dense.
        \end{corollary}

    \iflong
    \subsection{Meta-Theorems for FO and MSO}
        We can now state the main algorithmic consequences for temporal model checking.
        They imply that when a temporal property (such as existence of a connected component, of a matching, or of vertex-disjoint paths) can be expressed in an FO, respectively MSO, formula using the $\mathsf{X}$-encoding, then the corresponding computational problem is \FPT on nowhere dense, respectively bounded treewidth, graphs when parameterised by $\mathsf{X}$.
       
        \begin{theorem}
        \label{thm:temporal-mso-meta}
            \MCMSO on a \tg $\gcal=(G,\lambda)$ is \FPT when parameterised by
            (i)~$\twlifetime$, (ii)~$\twdegree$, (iii)~$\vim$, or (iv)~$\tim$.
        \end{theorem}
        \begin{proof}
            The statement follows by combining the MSO model-checking meta-theorem on bounded-treewidth Gaifman graphs (\Cref{thm:MSO-MC})
            with \Cref{lem:lifetime_enc-preserves-tw,lem:tdegree_enc-preserves-tw,lem:vim_enc-preserves-tw,lem:tim_enc-preserves-tw}.
        \end{proof}
        \begin{theorem}
            \MCFO on a \tg $\gcal=(G,\lambda)$ is \FPT when parameterised by
            (i)~$\lifetime$~if the footprint $G$ is nowhere dense, (ii)~$\tdegree$, (iii)~$\vim$, or (iv)~$\tim$.
        \end{theorem}
        \begin{proof}
            The statement follows by combining the FO model-checking meta-theorem on nowhere dense graph classes (\Cref{thm:FO-MC})
            with \Cref{lem:lifetime_enc-preserves-nwdense,lem:tdegree_enc-preserves-nwdense,lem:vim_enc-preserves-nwdense,lem:tim_enc-preserves-nwdense}.
        \end{proof}
        Note that for the parameters $\tdegree$, $\vim$, and $\tim$, no additional assumption on the footprint is required: bounded temporal degree directly yields a Gaifman graph of bounded maximum degree (which is nowhere dense), while bounded $\vim$ or $\tim$ implies bounded treewidth of the footprint (hence nowhere dense).
        In contrast, bounded lifetime alone does not imply nowhere denseness of the footprint: a clique in which all edges appear at time $1$ has lifetime~$1$ but is not nowhere dense.
         \fi       
\iflong
\section{Logic Cookbook: A Collection of Logical Formulas for Temporal Graph Properties}\sectionmark{Logic Cookbook}
\label{sec:cookbook}
    \newcommand{\LOGstatus}[1]{\textcolor{purple}{ -- Status: #1/3}}
    \newcommand{\LOGstatusthree}{\textcolor{LimeGreen}{ -- Status: to be read}}
    In this section we provide a lexicon of logical formulas in First or Second Order Logic for the four encodings. For any existing formulas, we provide references and unify using our framework.
    First, we give an overview of notation and conventions we employ.
    \subsection*{Notation and Conventions}
        \begin{description}
            \item[Variables.] For convenience, We use the following conventions for variables where possible:
                \begin{description}[font= \normalfont\sffamily, leftmargin=!, labelwidth=0pt]
                    \item[vertices:] $x,y,z,\dots,x_1,x_2,\dots$
                    \item[temporal edges:] $\varepsilon,\varepsilon_1,\varepsilon_2,\dots$
                    \item[vertex sets:] $X,Y,Z,\dots$
                    \item[temporal edge sets:] $\ecal',\ecal_1,\ecal_2,\dots$
                    \item[\vim bags:] $\vtimbag,\vtimbag_1,\vtimbag_2,\dots$
                    \item[\tim bags:] $\vtimbag,\vtimbag_1,\vtimbag_2,\dots$
                    \item[\vim bag sets:] $\vtimbagset,\vtimbagset_1,\vtimbagset_2,\dots$
                    \item[\tim bag sets:] $\vtimbagset,\vtimbagset_1,\vtimbagset_2,\dots$
                \end{description}
            \item[Shorthand quantification.]
                We use $\forall x\in V\,(\dots)$ as shorthand for
                $\forall x\,(V(x)\rightarrow \dots)$, where $V(x)$ is the unary relation for vertices in the universe. We apply this shortcut for each of the unary relations.
            \item[Static edges.]
                We use $uv\in E$ or $\edge(u,v)$ as shorthand for
                $\exists \eps\in \mathcal{E}:\,\inc(\eps,u)\wedge \inc(\eps,v)$. Similarly, to quantify over static edges, we may write $\exists E'\subseteq E, \,\forall \eps_1,\eps_2\in E', \,v,u\in V: \inc(\eps_1,v) \wedge \inc(\eps_2,v)\wedge \inc(\eps_1,u) \wedge \inc(\eps_2,u) \implies \eps_1=\eps_2$ since this forces at most one temporal edge per static edge.
            \item[Formula length and logic used.]
                An asterisk indicates that a formula has length $\ell$, but may contain longer subformula(e). Similarly, an asterisk on the fragment of logic used indicates that a formula is written in FO, but it could contain subformula(e) in MSO.
            \end{description}
    \setcounter{tocdepth}{3} 
    \localtableofcontents \label{toc:formulas}

    \subsection{Useful shortcuts}
    \label{subsec:usefule-shortcuts}
    \MSOSetCategory{Shortcuts}
    This section collects useful shortcut formulas. These formulas do not encode computational problems themselves, but serve as building blocks and help unify notation across encodings. For instance, although adjacency is expressed differently for individual encodings (see \Cref{subsec:sc-adjacency}), it can be referred to uniformly as~$\varphi_{\text{adj}}$ once expressability is established.

    \subsubsection{Static edges 
    } \label{subsec:static-edges}
         We use $uv\in E$ or $\edge(u,v)$ as shorthand for
                $\exists \eps\in \mathcal{E}:\,\inc(\eps,u)\wedge \inc(\eps,v)$. 
        Similarly, to quantify over static edges, we may write the following formula:
        
        \MSOFormula{edgeset}{quantifying static edges}
            {there exists a set of edges $E'$ in the underlying graph}
            {}
            {\lifetime, \tdegree, \vim, \tim}{constant}
            {\exists E'\subseteq E, \,\forall \eps_1,\eps_2\in E', \,v,u\in V: \\
            &\inc(\eps_1,v) \wedge \inc(\eps_2,v)\wedge \inc(\eps_1,u) \wedge \inc(\eps_2,u) \implies \eps_1=\eps_2}
            {MSO}{folklore}
        \noindent since this forces at most one temporal edge per static edge.
        
        \MSOFormula{sharededge}{shared static edge}
            {two temporal edges share an underlying edge}
            {\eps_1,\eps_2}
            {\lifetime, \tdegree, \vim, \tim}{constant}
            {\forall v\in V : \inc(\eps_1,v) \implies \inc(\eps_2,v)}
            {FO}{folklore}
        
    \subsubsection{Possible successor 
    } \label{subsec:sc-possuc}
         \input{Formulas/shortcuts/subsec--sc-pos_suc}

        \subsubsection{Temporal Adjacency 
        } \label{subsec:sc-adjacency}

\input{Formulas/shortcuts/subsec--sc-tem_adjacency}

        \subsubsection{Cardinalities
        }\label{subsec:sc-cardinalities}

\input{Formulas/shortcuts/subsec--sc-card}

        \subsubsection{Degree Requirements 
        } \label{subsec:sc-degree}

\input{Formulas/shortcuts/subsec--sc-degree_req}

        \subsubsection{Interval Demands 
        }\label{subsec:sc-interval}

\input{Formulas/shortcuts/subsec--sc-interval}

        \subsubsection{Static Connected Components 
        }\label{subsec:sc-staticCC}

\input{Formulas/shortcuts/subsec--sc-static_cc}

    \subsection{Temporal Paths, Walks, and Reachability  
    } \label{subsec:paths}
    \MSOSetCategory{Paths/Reachability}

\input{Formulas/shortcuts/subsec--sc-paths_walks}
        
    \subsection{Reach-based Problems}   
    This subsection summarizes temporal graph problems that can be expressed as refinements or restrictions of temporal reachability, using the path predicates introduced above.
            
            
            
        \subsubsection{Edge- and vertex-disjoint paths 
        }
        \label{subsec:disjoint}

\input{Formulas/ReachProblems/subsec--RP-disjoint}

        \subsubsection{Restless/bounded-waiting variants of reachability
        }
        \label{subsec:restless}

\input{Formulas/ReachProblems/subsec--RP-restless}

        \subsubsection{Temporally connected components 
        }
        \label{subsec:components}

\input{Formulas/ReachProblems/subsec--RP-TCC}

        \subsubsection{Temporal separators and cuts 
        }
        \label{subsec:separators}

\input{Formulas/ReachProblems/subsec--RP-separators} 
    
        \subsubsection{Temporal spanners 
        }
        \label{subsec:Spanners}

\input{Formulas/ReachProblems/subsec--RP-Spanner}  

        \subsubsection{Exploration (\lifetime) 
        }

\input{Formulas/ReachProblems/subsec--RP-Exploration}

    
    \subsection{Temporal Clique and Independent Set 
    }

\input{Formulas/subsec--Clique-IS}

    \subsection{Feedback Edge/Vertex Set 
    }
    \label{subsec:feedbackEdgecover}

\input{Formulas/subsec--FES_EdgeCover}

    \subsection{Temporal Colouring 
    }

\input{Formulas/subsec--Coloring}

    \subsection{Temporal Matchings 
    }

\input{Formulas/subsec--Matching}

    \subsection{Temporal Edge Cover, Vertex Cover, and Dominating Set 
    }
    \label{subsec:Cover}

\input{Formulas/subsec--Cover}

    
    %


    \subsubsection{A note on realisation and modification problems}

    We refer to problems in which we are given a static graph to label with times as a \emph{realisation} problem, and a temporal graph which we are required to add/remove (temporal) vertices or edges as a \emph{modification} problem. Realisation and modification problems have been expressed in MSO by capturing a property of the underlying graph which forces/forbids a temporal graph to have some (un)desired property~\cite{deligkas_minimizing_2023,enright_deleting_2021,kutner_temporal_2023,mertzios_temporal_2025}. 
            
            Note that existence of a labelling function can be encoded using a formula of length lifetime where the edges in the footprint graph are the union of $\lifetime$ many sets.
            \[\exists E_1,\ldots,E_{\Lambda}\subseteq E, \forall e\in E, \exists E_i : e\in E_i 
            \]
            We can also force these sets to form a partition if we require a simple temporal graph (a temporal graph such that each edge is active exactly once). This can be achieved by appending the following formula to the formula above.
            \[\big(e\in E_i \implies \bigwedge_{j\in 1,\ldots,\lifetime}(E_j=E_i \vee e\notin E_j)\big)\]
            
            Since modification problems in the literature require formulas which describe very specific properties of the footprint graph, we have not included them here. However, the temporal properties they aim to modify for may be found earlier in this section.

\fi

\section{Future Work}
\label{sec:conclusion}
    There are several directions for future research on the use of logic in the study of temporal graphs.

   With respect to temporal encodings, a natural direction is to identify further temporal parameters for which similar preservation and meta-theorems can be obtained. Possible candidates include the temporal feedback vertex number \cite{haag_feedback_2022}, the temporal cliquewidth \cite{enright_structural_2024}, or a combination of underlying cliquewidth with lifetime or temporal degree.

        In the work introducing the \timText width (\tim), Enright et al.~\cite{enright_families_2025} present two meta-algorithms that characterise the classes of computational problems that are fixed-parameter tractable with respect to \vim and \tim. For \vim, tractability is tied to a notion of \textit{temporal locality}, while for \tim this locality is further structurally restricted.
        This perspective is reminiscent of first-order logic which is inherently limited to see only local properties of a graph. An interesting question is whether these classes of problems admit a precise logical characterisation, \eg by a suitable fragment of FO or MSO, and how such a characterisation relates to the \vimText and \timText encodings.

    Finally, a broader direction is to investigate logical tools beyond FO and MSO meta-theorems, such as logical transductions, game-based techniques, and methods inspired by the Weisfeiler--Leman framework.
    Such tools may provide alternative ways to compare temporal graph classes and to reason about their structural and algorithmic properties.
    Recent work has begun to study isomorphisms for temporal graphs under reachability notions~\cite{casteigts_simple_2024,doring_simple_2025}; an intriguing next step is to explore whether transductions or related techniques can be 
    extended,
     used to further compare temporal graph classes,
    in particular with respect to structural parameters and the computational complexity of temporal graph problems.


\bibliography{msoTG, references}
\end{document}

%% file: probenv.tex
\newenvironment{tightcenter}
 {\parskip=0pt\par\nopagebreak\centering}
 {\par\noindent\ignorespacesafterend}

\usepackage{tikz}
\usetikzlibrary{calc}
\usepackage{xargs}
\usepackage{xifthen}
\usepackage{framed}

\usepackage{ctable}

\newlength{\RoundedBoxWidth}
\newsavebox{\GrayRoundedBox}
\newenvironment{GrayBox}[1]%
   {\setlength{\RoundedBoxWidth}{\textwidth-10.5ex}
    \def\boxheading{#1}
    \begin{lrbox}{\GrayRoundedBox}
       \begin{minipage}{\RoundedBoxWidth}%
   }{%
       \end{minipage}
    \end{lrbox}%
    \begin{tightcenter}%
    \begin{tikzpicture}%
       \node(Text)[draw=black!90,fill=white,rounded corners,%
             inner sep=2ex,text width=\RoundedBoxWidth]%
             {\usebox{\GrayRoundedBox}};
        \coordinate(x) at (current bounding box.north west);
        \node [draw=white,rectangle,inner sep=3pt,anchor=north west,fill=white] 
        at ($(x)+(6pt,.75em)$) {\boxheading};
    \end{tikzpicture}
    \end{tightcenter}\vspace{0pt}%
    \ignorespacesafterend
}    

\newenvironment{problem}[2][]{\noindent\ignorespaces%
                                \FrameSep=8pt%
                                \parindent=0pt%
                \vspace*{-.5em}
                \ifthenelse{\isempty{#1}}{%
                  \begin{GrayBox}{\textsc{#2}}%
                }{%
                }
                \newcommand\Prob{{Problem:}}%
                \newcommand\Input{{Input:}}%
                          
                \begin{tabular*}{\textwidth}{@{\hspace{.1em}} >{\itshape} p{1.2cm} p{0.85\textwidth} @{}}%
            }{
                \end{tabular*}%
                \end{GrayBox}%
                \vspace*{-.5em}
                \ignorespacesafterend
            } 

%% file: _intro-mi-2.tex

Many networks of interest, such as
public-transport systems or communication and interaction networks, consist of time-sensitive connections.
A standard abstraction of such networks is a \emph{temporal graph} $\gcal=(V,E,\lambda)$, that is, a static graph $G=(V,E)$ as \textit{footprint} together with an edge labelling $\lambda$ assigning to each edge the time steps at which it is available.
This model strictly generalises static graphs, which are recovered as the special case in which all edges are available at all times.
While temporal graphs are more expressive, this comes at a computational costs:
many problems that are polynomial-time solvable on static graphs are computationally hard on temporal graphs, motivating an extensive line of work in parameterised complexity. 

The core question in parameterised complexity is whether a problem is fixed-parameter tractable (\FPT) when parameterised by some measure that captures relevant properties of the input. 
For static graphs, one of the most studied 
parameters is \emph{treewidth} (\tw).
Two standard approaches to obtain \FPT algorithms parameterised by treewidth are:
    (i) developing a dynamic program (see \cite[Section 7.3]{cygan_parameterized_2015}), or (ii) expressing the problem in a suitable \textit{logic} and applying a logical \emph{meta-theorem}.

\textit{First-order (FO)} and \textit{monadic-second order (MSO)} logic provide formal languages for expressing combinatorial problems on graphs.
Logical meta-theorems unify results by establishing tractability not only for individual problems, but for entire classes of problems on entire classes of graphs, often based on the fragment of logic required.
The most prominent example is Courcelle’s theorem~\cite{courcelle_graph_1990}, which implies that every property definable in MSO logic can be decided in linear time on graphs of bounded treewidth. On sparse graphs, FO logic admits analogous guarantees: Frick and Grohe~\cite{frick_deciding_2001} established tractability of FO model checking on planar graphs, and Grohe, Kreutzer, and Siebertz~\cite{grohe_deciding_2017} extended this to the much broader nowhere dense graph class which includes planar graphs and graphs of bounded treewidth, degree, or expansion.
For a detailed background on logic and algorithmic meta-theorems we refer to the surveys~\cite{grohe_logic_2008,kreutzer_algorithmic_2011}.

    For temporal graphs, applications of logical meta-theorems so far exclusively use MSO formulations to obtain \FPT results under bounded treewidth combined with lifetime (largest time label) or temporal degree (number of temporal edges incident to one vertex) \cite{davot2025parameterised,deligkas_parameterized_2025,enright_deleting_2021,enright_counting_2025,fomin_as_2020,haag_feedback_2022,kutner_temporal_2023,zschoche_complexity_2020}.
    In contrast, the FO meta-theorem route on planar or nowhere dense graphs has not yet been considered for temporal graphs.
    Moreover, it is still common to obtain \FPT results via problem-specific dynamic programs with long technical proofs (see, e.g., \cite{chakraborty_algorithms_2024,herrmann_timeline_2025}), even in settings where a logical approach could yield a simpler and more modular proof.
    Here, we provide a unifying framework for encoding temporal graphs along with example associated formulas. We also extend existing applications of logical meta-theorems to temporal graphs beyond bounded treewidth by enabling applications of FO meta-theorems for nowhere-dense footprints.

    Beyond proof techniques, a central question is which parameters meaningfully capture temporal and structural complexity. 
    Recently, two new parameters have been introduced that are sensitive to the \emph{order} of snapshots and thus capture temporal structure more intrinsically: vertex-interval-membership width (\vim) and tree-interval-membership width (\tim). 
    Enright et al.~\cite{enright_families_2025} study algorithmic tractability under these parameters, and give a procedural meta-algorithm for generating dynamic programs on a restricted class of problems: this is not a meta-theorem in the logical sense. 
    However, since bounded \vim\ and \tim\ bound the treewidth of the footprint, question arises whether it may be possible to prove a logical meta-theorem for these 
    measures.
    We answer this affirmatively: we show that \vim\ and \tim\ admit encodings within our framework that preserve the structure needed for FO/MSO meta-theorems.

%% file: _sec--related-work.tex

The first work in logic that bears resemblance to temporal graph encodings is by Arnborg and Lagergren~\cite{arnborg_easy_1991}, who provided an alternative proof of Courcelle’s MSO meta-theorem and, in particular, considered edge-labelled static graphs. 
This perspective was later adopted by Zschoche et al.~\cite{zschoche_complexity_2020} to obtain an \FPT algorithm for \textsc{Temporal $(s,z)$-Separator} parameterised by \twlifetime, introducing what we call the lifetime encoding.

Independently, Mans and Mathieson~\cite{mans_treewidth_2014} 
studied different models of dynamic graphs (including temporal graphs) of bounded lifetime and proposed an encoding that was claimed to extend the MSO and FO meta-theorems of Courcelle and Frick--Grohe to this setting.  However, their encoding contains a fundamental flaw, which was later identified and shown to be irreparable in the survey by Fomin et al.~\cite{fomin_as_2020}. That survey gives a comprehensive exploration of possible notions of treewidth for temporal graphs (since no singular temporal treewidth definition exists), discussing both the lifetime and degree encodings in detail. 
The degree encoding was introduced by Enright et al.~\cite{enright_deleting_2018}; they used it to show that \textsc{Temporal Reachability Edge Deletion} is in \FPT with respect to \twdegree and solution size combined.


Subsequent work also uses an MSO approach: 
Haag et al.~\cite{haag_feedback_2022} prove \FPT by \twlifetime for \textsc{Temporal Feedback Edge/Connection Set};
Kutner and Larios-Jones~\cite{kutner_temporal_2023} show \FPT by \twlifetime for \textsc{Temporal Reachability Dominating Set};
Enright et al.~\cite{enright_counting_2025} show \FPT by \twdegree for \textsc{Counting Temporal Paths} via a counting variant of MSO;
Deligkas et al.~\cite{deligkas_parameterized_2025} show \FPT by \twlifetime, and \twdegree for \textsc{Open} and \textsc{Closed Temporal Connected Component}; and 
Mertzios et al.~\cite{mertzios_temporal_2025} use Courcelle-type reasoning over a \emph{static} MSO formulation on the footprint for \textsc{Temporal Graph Realization with Fixed Stretch}.


%% file: Formulas/shortcuts/subsec--sc-pos_suc.tex

Two edges can be consecutive in a temporal path only if they share a common endpoint and the first appears earlier in time than the second. In the degree encoding, this notion is directly represented by a binary relation capturing overlapping \textit{temporal} edges whose time labels are ordered; the formula \msovarphi{psucDEG} is provided for uniformity of subformulas.
For the lifetime, VIM, and TIM encoding, an analogous notion can be expressed by first-order or monadic second-order formulas, depending on whether time is explicitly available in the universe (lifetime) or only implicitly encoded via bags (VIM, TIM).

\MSOFormula{<bagVIM}{$\vtimbag_1$ is strictly earlier than $\vtimbag_2$ (\vim
)}
    {the time step of $\vtimbag_1$ is smaller than the time step of $\vtimbag_2$}
    {\vtimbag_1,\vtimbag_2}
    {\vim
    }
    {constant}
    {
        \exists P\subseteq\vtimbagset\colon \Big(
        \vtimbag_1\in P \wedge \vtimbag_2\in P\ \wedge \\&
        \big( \forall\vtimbag\in P\setminus\{\vtimbag_1\}, \exists \vtimbag'\in P\colon \bagbefore(\vtimbag',\vtimbag) \big) \wedge\big( \forall\vtimbag\in P\setminus\{\vtimbag_2\}, \exists \vtimbag'\in P\colon \bagbefore(\vtimbag,\vtimbag') \big)
        \Big)
    }
    {MSO}{folklore}

We cannot use this formula as is for the TIM encoding as the path between two bags at the same time in the decomposition may not be strictly increasing or decreasing. However, if we mandate that the two bags contain a shared vertex, the bags containing them must form a strictly increasing or decreasing path in the TIM decomposition.

\MSOFormula{<bagTIM}{$\vtimbag_1$ is strictly earlier than $\vtimbag_2$ (\tim
)}
    {the time step of $\vtimbag_1$ is smaller than the time step of $\vtimbag_2$, and the bags share a vertex}
    {\vtimbag_1,\vtimbag_2}
    {\tim
    }
    {constant}
    {   \exists v\in V, 
        \exists P\subseteq\vtimbagset\colon \Big(
        \vtimbag_1\in P \wedge \vtimbag_2\in P \wedge \bag(v,\vtimbag_1) \wedge \bag(v,\vtimbag_2)\ \wedge \\&
        \big( \forall\vtimbag\in P\setminus\{\vtimbag_1\}, \exists \vtimbag'\in P\colon \bagbefore(\vtimbag',\vtimbag) \big) \wedge\big( \forall\vtimbag\in P\setminus\{\vtimbag_2\}, \exists \vtimbag'\in P\colon \bagbefore(\vtimbag,\vtimbag') \big)
        \Big)
    }
    {MSO}{folklore}
\msoalias{<bagIM}{<bagVIM}
    
Note that we can modify these formulas to allow nonstrict paths by allowing that either $\pi$ and $\pi'$ are the same bag or $\bagbefore(\pi,\pi')$.

\MSOFormula{psucLIFE}{possible successor (\lifetime)}
    {an ordered pair of temporal edges could be consecutive in a temporal path}
    {\eps_1, \eps_2}
    {\lifetime}{constant}
    {\exists v\in V, \exists t_1,t_2\in L\colon \Big(\pres(\eps_1,t_1) \wedge \pres(\eps_2,t_2) \wedge
    \inc(\eps_1,v) \wedge \inc(\eps_2,v)\ \wedge\\&
    t_1 \timebefore t_2 \Big)}
    {FO}{folklore}
    
\MSOFormula{psucIM}{possible successor (\vim
)}
    {an ordered pair of temporal edges could be consecutive in a temporal path}
    {\eps_1, \eps_2}
    {\vim, \tim}{constant}
    {\exists v\in V, \exists \vtimbag_1, \vtimbag_2\in\vtimbagset\colon \Big( \pres(\eps_1,\vtimbag_1) \wedge \pres(\eps_2,\vtimbag_2) \wedge \inc(\eps_1,v)\wedge \inc(\eps_2,v)\ \wedge \\&
        \msovarphi{<bagIM}(\vtimbag_1,\vtimbag_2)
    \Big)
    }
    {FO$^*$}{folklore}

    Here $\msovarphi{<bagIM}$ is a substitute for either $\msovarphi{<bagVIM}$ or $\msovarphi{<bagTIM}$.
    

\MSOFormula{psucDEG}{possible temporal successor (\tdegree)}
    {an ordered pair of temporal edges could be consecutive in a temporal path}
    {\eps_1, \eps_2}
    {\tdegree}{constant}
    {\psuc(\eps_1,\eps_2)}
    {FO}{folklore}

\msoalias{psuc}{psucIM}
\msoalias{psuc}{psucDEG}
\msoalias{psuc}{psucLIFE}

%% file: Formulas/shortcuts/subsec--sc-tem_adjacency.tex
Temporal adjacency captures the existence of a temporal edge from~$x$ to~$y$ (at some time $t$). Its formulation depends on the available information of the encoding.

\MSOFormula{TadjLIFE}{temporal adjacency (\lifetime)}
    {there exists an edge between $x$ and $y$ at time $t$}
    {x,y,t}
    {\lifetime}{constant}
    {\exists \eps\in \mathcal{E}\colon \big(\inc(\eps,x) \wedge \inc(\eps,y) \wedge \Sedgetime(\eps,t)\big)}
    {FO}{folklore}
\noindent It is immediate to write a formula which is true if and only if $x$ and $y$ are adjacent at any time $\varphi_\text{adjLIFE}(x,y)=\exists t\in T\colon \msovarphi{TadjLIFE}(x,y,t)$.

For the degree encoding, there is no notion of specific times, so we can only express that there exists some temporal edge connecting $x$ to $y$.
\MSOFormula{TadjDEG}{temporal adjacency (\tdegree)}
    {there exists a time $t$ such that there is an edge between $x$ and $y$ at time $t$}
    {x,y}
    {\tdegree}{constant}
    {\exists\eps\in\ecal\colon\big(\inc(\eps,x) \wedge \inc(\eps,y)\big)}
    {FO}{folklore}

\MSOFormula{TadjVIM}{temporal adjacency (\vim)}
    {there exists and edge between $x$ and $y$ at time $t$}
    {x,y, \vtimbag_t}
    {\vim}{constant}
    {\exists \eps\in \mathcal{E}\colon\big(\inc(\eps,x) \wedge \inc(\eps,y)\wedge\Sedgetime(\eps,\vtimbag_t)\big)}
    {FO}{folklore}
\noindent It is immediate to write a formula which is true if and only if $x$ and $y$ are adjacent at any time $\varphi_\text{adjVIM}(x,y)=\exists \vtimbag_t\in \vtimbagset\colon \msovarphi{TadjVIM}(x,y,\vtimbag_t)$.
    
For the \tim encoding, there is no notion of specific times, in particular since we cannot check if two bags belong to the same time. So we can only express that there exists some temporal edge connecting $x$ to $y$.
\MSOFormula{TadjTIM}{temporal adjacency (\tim)}
    {there exists an edge between $x$ and $y$ at some time}
    {x,y}
    {\tim}{constant}
    {\exists \eps\in \mathcal{E}, \vtimbag\in \vtimbagset\colon\big(\inc(\eps,x) \wedge \inc(\eps,y)\wedge\Sedgetime(\eps,\vtimbag)\big)}
    {FO}{folklore}



%% file: Formulas/shortcuts/subsec--sc-card.tex

Cardinality measures the size of a set and allows for comparisons between set sizes.
Often we don't need the below formula (and the dependence on $k$) because an extension of Courcelle's theorem allows us to encode the corresponding optimisation problem without dependence on solution size.


\MSOFormulaD{card}
    {cardinality $k$ \cite{kutner_temporal_2023}}
    {the set $X$ has cardinality at least $k$}
    {k}
    {\lifetime, \tdegree,\vim,\tim}{k}
    {(X)=\exists x_1,\ldots,x_k \in X\colon \bigwedge_{1\leq i<j\leq k}x_i\neq x_j}
    {FO}{}


%% file: Formulas/shortcuts/subsec--sc-degree_req.tex
Degree requirements specify constraints on the number of temporal edges incident to a given vertex within a specified set of temporal edges.
Such predicates are frequently used to enforce local structural properties of edge sets, for instance when characterizing temporal paths.

\MSOFormula{deg=0}{degree 0}
    {vertex $x$ has no incident edge in the set  $P$ of temporal edges}
    {x,P}
    {\lifetime, \tdegree, \vim, \tim}{constant}
    {\forall \eps\in P \colon \neg\inc(\eps,x)}
    {FO}{folklore}

\MSOFormulaD{deg}{degree $k$}
    {vertex $x$ has exactly $k$ incident edges in the set $P$ of temporal edges}
    {k}
    {\lifetime, \tdegree, \vim, \tim}{$k$}
    {(x,P) = 
        \exists \eps_1, \dots,\eps_k\in P, \forall \eps'\in P \colon \Big(\bigwedge_{1\leq i<j\leq k}\eps_i\neq \eps_j \, \wedge\\& \bigwedge_{i\in[k]}\inc(\eps_i,x) \wedge \big(\inc(\eps',x)\implies \bigvee_{i\in[k]}\eps' =\eps_i\big)\Big)
    }
    {FO}{folklore}
\msoalias{deg=1}{deg}
\msoalias{deg=2}{deg}

%% file: Formulas/shortcuts/subsec--sc-interval.tex
Interval demands constrain the temporal distance between time steps or edges.

\MSOFormulaD{LintLIFE}{two time steps at most $\ell$ apart (\lifetime)}
    {time steps $t_1$ and $t_2$ are at most $\ell$ apart}
    {\ell}{\lifetime}
    {$\ell$}
    {(t_a,t_b) = \exists t_1,\ldots,t_\ell \in L\colon \Big( t_1=t_a \wedge t_\ell=t_b\ \wedge\\&   \forall t\in L\colon \big( t_1\timebefore t \wedge t\timebefore t_\ell)\implies \bigvee_{1\leq i\leq \ell} t_i=t
    \big)\Big)}
    {FO}{folklore}

\MSOFormulaD{TEintLIFE}{two temporal edges at most $\ell$ time steps apart (\lifetime)}
    {temporal edges $\eps_a$ and $\eps_b$ are at most $\ell$ time steps apart}
    {\ell}{\lifetime}
    {constant$^*$}
    {(\eps_a,\eps_b) = 
    \forall t_a,t_b\in L\colon \big( (\Sedgetime(\eps_a,t_a) \wedge \Sedgetime(\eps_b,t_b) )\implies {\msovarphi{LintLIFE}}_\ell(t_a,t_b) \big)}
    {FO}{folklore}


\MSOFormulaD{BintIM}{two bags at most $\ell$ time steps apart (\vim
)}
    {bags $\vtimbag_a$ and $\vtimbag_b$ are at most $\ell$ time steps apart}
    {\ell}
    {\vim
    }
    {$\ell$}
    {(\vtimbag_a,\vtimbag_b) = 
        \exists \vtimbag_0,\ldots,\vtimbag_{\ell} \in \vtimbagset \colon \Big(
        \vtimbag_0=\vtimbag_a \wedge \vtimbag_\ell=\vtimbag_b \wedge \\&
        \bigwedge_{0\leq i\leq \ell-1} (\bagbefore(\vtimbag_{i},\vtimbag_{i+1})\vee \vtimbag_{i} = \vtimbag_{i+1}))
        \Big)
    }
    {FO}{novel}

\MSOFormulaD{TEintIM}{two temporal edges at most $\ell$ apart (\vim
)}
    {temporal edges $\eps_a$ and $\eps_b$ are at most $\ell$ time steps apart}
    {\ell}
    {\vim
    }
    {constant$^*$}
    {(\eps_a,\eps_b) = 
        \forall \vtimbag_a,\vtimbag_b\in\vtimbagset\colon \big( (\Sedgetime(\eps_a,\vtimbag_a) \wedge \Sedgetime(\eps_b,\vtimbag_b)) \implies {\msovarphi{BintIM}}_{\ell}(\vtimbag_a,\vtimbag_b)\big)
    }
    {FO}{folklore}

\msoalias{TEint}{TEintIM}
\msoalias{TEint}{TEintLIFE}


\MSOFormulaD{TEintTIM}{two temporal edges share an endpoint and are within a given interval of one another (\tim)}
    {$\eps_1$, $\eps_2$ occur within $\ell$ of each other}
    {\ell}
    {\tim
    }
    {$\ell$}
    {(\eps_1,\eps_2) = 
        \exists x\in V, \exists \vtimbag_1,\ldots,\vtimbag_l \in \vtimbagset,  \forall \vtimbag\in \vtimbagset:
    \\&
    \Big(\ \inc(\eps_1,x)\wedge\inc(\eps_2,x)\wedge \bagbefore(\vtimbag_1,\vtimbag) \wedge \bagbefore(\vtimbag,\vtimbag_l)
      \implies \bigvee_{1\leq i\leq \ell} \vtimbag_i=\vtimbag \ 
     \wedge \\& \big(\Sedgetime(e_1,\vtimbag_1) \wedge \Sedgetime(e_2,\vtimbag_l)\big)
    \Big)
    }
    {FO}{folklore}
\msoalias{TEint}{TEintTIM}

%% file: Formulas/shortcuts/subsec--sc-static_cc.tex
Static connected components capture connectivity properties within individual snapshots or within the static footprint of a temporal graph.
In contrast to temporal connected components, static components are defined by ignoring temporal order and considering only the presence of edges at a fixed time or across all times.

\paragraph{Static connected components over vertices in a snapshot}
Given a snapshot $G_t$, a vertex set induces a static connected component if it is connected via edges present at time~$t$ and is maximal with respect to this property.
The following formulas express static connected components both as vertex sets and as a binary relation between vertices.

\MSOFormula{XccLIFE}{set of vertices induces static connected component at a given time (\lifetime) \cite{haag_feedback_2022,kutner_temporal_2023}}
    {$X$ induces a static connected component in $G_t$}
    {X,t}
    {\lifetime}{constant}
    { \forall Y \subset X \colon \Big( Y\neq\emptyset \implies \big(\exists y\in Y,\exists x\in X\setminus Y, \exists e\in E \colon \\&( \inc(e,x) \wedge \inc(e,y)\, \wedge \pres(e,t) ) \big) \Big)
    }
    {MSO}{\cite{kutner_temporal_2023}}

\MSOFormula{ccLIFE}{$u$ and $v$ are in the same static connected component at a given time (\lifetime) \cite{haag_feedback_2022,kutner_temporal_2023}}
    {$u,v$ are in the same static connected component in $G_t$}
    {u,v,t}
    {\lifetime}{constant}
    { \exists X\subset V \colon \big( u\in X\wedge v\in X \wedge \msovarphi{XccLIFE}(X,t) \big)
    }
    {MSO}{\cite{kutner_temporal_2023}}


\MSOFormula{XccVIM}{set of vertices induces static connected component at a given time (\vim)}
    {$X$ induces a static connected component in $G_t$}
    {X,\vtimbag_t}
    {\vim}{constant}
    { \forall Y \subset X\colon \Big( Y\neq\emptyset \implies \big( \exists y\in Y,\exists x\in X\setminus Y, \exists e\in E \colon \\
    &( \inc(e,x)\wedge\inc(e,y)\wedge \pres(e,\vtimbag_t) ) \big) \Big)
    }
    {MSO}{novel}
    
\MSOFormula{ccVIM}{$u$ and $v$ are in the same connected component  at a given time (\vim)}
    {$u,v$ are in the same static connected component in $G_t$}
    {u,v,\vtimbag_t}
    {\vim}{constant}
    { \exists X\subset V \colon \big( u\in X\wedge v\in X\wedge \msovarphi{XccVIM}(X,\vtimbag_t) \big)
    }
    {MSO}{novel}

\MSOFormula{XccTIM}{there exists a snapshot such that set of vertices induces static connected component in that snapshot (\tim)}
    {there exists a time $t$ such that $X$ induces a static connected component in $G_t$}
    {X}
    {\vim}{constant}
    { \exists \vtimbag\in \vtimbagset, \forall Y \subset X \colon \Big(Y\neq\emptyset \implies \exists y\in Y,\exists x\in X\setminus Y, \exists \eps\in \mathcal{E}\\
    &\big( \inc(\eps,x)\wedge\inc(\eps,y)\wedge \pres(\eps,\vtimbag)\big)
    }
    {MSO}{novel}
    
\MSOFormula{ccTIM}{there exists a snapshot such that $u$ and $v$ are in the same connected component of the snapshot (\tim)}
    {there exists a time $t$ such that $u,v$ are in the same connected component in $G_t$}
    {u,v}
    {\tim}{constant}
    { \exists X\subset V \colon \big( u\in X\wedge v\in X\wedge \msovarphi{XccTIM}(X) \big)
    }
    {MSO}{novel}

\paragraph{Static connected components induced by temporal edges in the footprint}

Every vertex set $X$ that contains the root $r$ and is closed under stepping along edges of~$P_\ecal$ must also contain $x$. This is the standard MSO characterization of reachability.
\MSOFormula{FPreach}{$r$ and $x$ are in the same static connected component in the footprint induced by $P_\ecal$~\cite{haag_feedback_2022}}
    {$x$ is reachable from $r$ in the footprint subgraph induced by $P_\ecal$}
    {r,x,P_\ecal}
    {\lifetime, \tdegree, \vim, \tim}{constant}
    {
        \forall X\subseteq V\colon \\&
        \Big(r\in X \wedge\ 
        \forall v\in X, \forall \eps\in P_\ecal\colon \big( \inc(\eps,v)\implies \exists w\in X\setminus\{v\}\colon \inc(\eps,w) \big) \Big)
        \implies x\in X
        \Big)
    }
    {MSO}{folklore}
    
Given a set of temporal edges, a static connected component is induced if the vertices incident to these edges are connected in the footprint graph obtained by ignoring temporal information.

\MSOFormula{TEcc}{set of temporal edges induces a static connected component in the footprint~\cite{haag_feedback_2022}}
    {the temporal edge set $P_\ecal\subseteq\ecal$ induces a connected subgraph in the static footprint}
    {P_\ecal}
    {\lifetime, \tdegree, \vim, \tim}{constant}
    {
        \exists r\in V, \exists \eps\in P_\ecal\colon \Big( \inc(\eps,r)\ \wedge\ 
            \forall x\in V\colon
                \big(\exists \eps\in P_\ecal\colon \inc(\eps,x)\big)
                \implies \msovarphi{FPreach}(r,x, P_\ecal)
        \Big)
    }
    {MSO}{folklore}

%% file: Formulas/shortcuts/subsec--sc-paths_walks.tex
Reachability is a fundamentally global graph property. While first-order logic (FO) can express the existence of a path of fixed length $k$ between two vertices, the length of the corresponding FO formula grows with $k$. It is a classical result that FO is not able to express reachability, \ie the existence of a path of arbitrary length between two vertices; see for example \cite[Chapter 3.6]{libkin_elements_2004}.

In contrast, monadic second-order logic (MSO) can express reachability by quantifying over vertex or edge sets and enforcing the connectivity properties of a path in a single formula of constant size. This strict separation between FO and MSO with respect to reachability underlies the use of MSO as a logical framework for path-based properties in graph algorithms and meta-theorems. 

For the four temporal encodings the reachability can be expressed using MSO, shown in \Cref{subsubsec:MSO-reachability}. 
For lifetime encoding and non-strict temporal reachability the FO impossibility directly translates since otherwise one could express static reachability of a static graph $G$ via a temporal graph $(G,{e,1 \colon e\in E})$ and then using the temporal formula. For strict temporal reachability, this is not directly true and for lifetime encoding one can even express in FO. 

\subsubsection{MSO formulas for temporal reachability 
} \label{subsubsec:MSO-reachability}
Temporal reachability is expressed by quantifying a set $P$ of temporal edges and enforcing local degree conditions: the designated endpoints $u$ and $v$ are required to have degree~$1$ in $P$, every other vertex incident at an edge in $P$ must have degree~$2$ in $P$, and all remaining vertices have degree~$0$. These conditions guarantee that $P$ forms a simple $u$--$v$ path in the static footprint of the temporal graph. This static path is lifted to a temporal one, by an additional condition that the temporal edges in $P$ can be traversed in temporal order using the possible successor relation/formula. 

Note that the following formulas work for every encoding by using the appropriate helper formulas from \Cref{subsec:sc-degree} (degree requirements) and \Cref{subsec:sc-possuc} (possible successor).
\MSOFormula{PEpath}{temporal edges $P_\ecal$ form temporal path from $u$ to $v$}
    {set of temporal edges $P_\ecal\subseteq \ecal$ forms a temporal path from $u$ to $v$}
    {P_\ecal,u,v}
    {\lifetime, \tdegree, \vim, \tim}{constant}
    {
        \exists \varepsilon_s,\varepsilon_z\in P_\ecal\colon \Big(
        \inc(\varepsilon_s,u) \wedge \inc(\varepsilon_z,v) \wedge 
        \msovarphi{deg=1}(u,P_\ecal) \wedge \msovarphi{deg=1}(v,P_\ecal) \ \wedge \\&  
        \forall \varepsilon \in P_\ecal\setminus\{\varepsilon_s,\varepsilon_z\}, \forall w\in V\colon \big( \inc(\varepsilon,w) \implies \msovarphi{deg=2}(w,P_\ecal) \big) \,\wedge \msovarphi{TEcc}(P_\ecal)\\& 
        \forall \varepsilon \in P_\ecal,\, \exists \varepsilon_j\in P_\ecal \colon  \msovarphi{psuc}(\varepsilon_i,\varepsilon_j)
         \Big)
     }
    {FO}{folklore}

\MSOFormula{PVpath}{vertices $P_V$ form temporal path from $u$ to $v$}
    {set of vertices $P_V\subseteq V$ forms a temporal path from $u$ to $v$}
    {P_V,u,v}
    {\lifetime, \tdegree, \vim, \tim}{constant}
    {
        \exists P_\ecal\subseteq\ecal\colon \Big(\msovarphi{PEpath}(u,v,P_\ecal) \wedge\\& \forall \varepsilon\in P_\ecal, \exists v_1,v_2\in P_V \colon \big( v_1\neq v_2 \wedge \inc(\varepsilon,v_1) \wedge \inc(\eps,v_2) \big)
        \Big)
     }
    {MSO}{folklore}
    
\MSOFormula{path}{existence of temporal path from $u$ to $v$}
    {there exists a temporal path from $u$ to $v$}
    {u,v}
    {\lifetime, \tdegree, \vim, \tim}{constant}
    {
        \exists P_V\subseteq V \colon \msovarphi{PVpath}(P_V,u,v)}
    {MSO}{}

\paragraph{Restricted temporal paths}
Restricted variants of temporal reachability are obtained by additionally constraining the allowed vertices or edges to be used in the path.
\MSOFormula{pathV}{existence of temporal path restricted to vertices $X$}
    {there exists a temporal path from $u$ to $v$ using only vertices from $X$}
    {u,v,X}
    {\lifetime, \tdegree, \vim, \tim}{constant}
    {\exists P_V\subseteq X \colon  \msovarphi{PVpath}(u,v,P_V)}
    {MSO}{folklore}

\MSOFormula{pathSE}{existence of temporal path restricted to static edges $X_E$}
    {there exists a temporal path from $u$ to $v$ using only static edges from $X_E$}
    {u,v,X_E}
    {\lifetime, \tdegree, \vim, \tim}{constant}
    {\exists P_E\subseteq \mathcal{E}, \forall \eps\in P_E, \exists \eps'\in X_E \colon \msovarphi{sharededge}(\eps,\eps') \wedge \msovarphi{PEpath}(u,v,P_E)}
    {MSO}{folklore}

    Note that here we use a set of temporal edges as a stand-in for static edges by the shortcuts \msovarphi{edgeset} and \msovarphi{sharededge}.

\MSOFormula{pathTE}{existence of temporal path restricted to temporal edges $X_{\ecal}$}
    {there exists a temporal path from $u$ to $v$ using only temporal edges from $X_{\ecal}$}
    {u,v,X_{\ecal}}
    {\lifetime, \tdegree, \vim, \tim}{constant}
    {\exists P_{\ecal}\subseteq X_{\ecal} \colon \msovarphi{PEpath}(u,v,P_{\ecal})}
    {MSO}{folklore}

\paragraph{Special case: FO reachability for \lifetime in strict \tgs 
}
    Every temporal path in a strict \tg can use at most \lifetime\ edges as it can take at most one edge per time step. In the lifetime encoding this permits an explicit “unrolling” of paths and walks in a formula of length \lifetime which is possible in  FO.
    \MSOFormula{VwalkLIFE}{temporal walk (\lifetime)}
        {multiset of vertices $x_1,\ldots,x_{\Lambda}$ forms a temporal walk}
        {x_1,\ldots,x_{\lifetime}}
        {\lifetime}{$\Lambda$}
        {
            \bigwedge_{t=0}^{\lifetime-1}
            \big( x_t=x_{t+1} \vee \msovarphi{TadjLIFE}(x_t,x_{t+1},t)\big)
            }
        {FO}{novel}
    
    
    \MSOFormula{SVpathLIFE}{strict temporal path vertex set (\lifetime)}
        {set of vertices $x_1,\ldots,x_{\Lambda}$ forms a strict temporal path}
        {x_1,\ldots,x_{\lifetime}}
        {\lifetime}{$\Lambda$}
        {
            \bigwedge_{t=0}^{\lifetime-1}
            \big( x_t=x_{t+1} \vee (\msovarphi{TadjLIFE}(x_t,x_{t+1},t)\wedge \bigwedge_{i=0}^{t} x_i\neq x_{t+1})\big)
        }
        {FO}{folklore}
    
    \MSOFormula{SpathLIFE}{existence of strict temporal path between two vertices (\lifetime) \cite{kutner_temporal_2023}}
        {there exists a strict temporal path from $u$ to $v$}
        {u,v}
        {\lifetime}{$\Lambda$}
        {\exists x_0,\dots,x_{\lifetime}\in V \colon
            \big(x_0=u \wedge x_{\lifetime}=v \wedge
            \msovarphi{SVpathLIFE}(x_1,\ldots,x_{\Lambda})\big)}
        {FO}{folklore}
    

%% file: Formulas/ReachProblems/subsec--RP-disjoint.tex
Disjointness constraints strengthen temporal reachability by requiring multiple vertex- or edge-disjoint temporal paths between two vertices.


\MSOFormula{Edisjointpath}{existence edge-disjoint path}
    {there exist two temporal-edge-disjoint temporal paths from $u$ to $v$}
    {u,v}
    {\lifetime, \tdegree, \vim, \tim}{constant}
    {
        \exists P_1,P_2\subseteq \ecal \colon \Big(     \msovarphi{PEpath}(u,v,P_1) \wedge \msovarphi{PEpath}(u,v,P_2)\, \wedge \\&
        \forall \eps\in P_1\colon \eps\notin P_2 \wedge \forall \eps\in P_2\colon \eps\notin P_1 \Big)}
    {MSO}{novel}

\MSOFormula{Vdisjointpath}{existence vertex-disjoint path}
    {there exist two vertex-disjoint temporal paths from $u$ to $v$}
    {u,v}
    {\lifetime, \tdegree, \vim, \tim}{constant}
    {
        \exists P_1,P_2\subseteq V \colon \Big( \msovarphi{PVpath}(u,v,P_1) \wedge \msovarphi{PVpath}(u,v,P_2)\, \wedge \\&
        \forall v\in P_1\colon v\notin P_2 \wedge \forall v\in P_2\colon v\notin P_1 \Big)
        }
    {MSO}{}

%% file: Formulas/ReachProblems/subsec--RP-restless.tex
Restless temporal paths additionally bound the waiting time between consecutive temporal edges by a parameter~$\delta$. This can be captured by adjusting the possible successor formulas from \Cref{subsec:sc-possuc} by a $\delta$-guarded variant that only relates two temporal edges if they are consecutive along the path and the time gap between them is at most~$\delta$.

\MSOFormulaD{deltapsucLIFE}{possible $\delta$-successor (\lifetime)}
    {an ordered pair of temporal edges could be consecutive in a temporal path with waiting time at most $\delta$}
    {\delta}
    {\lifetime}{constant$^*$}
    {(\eps_1,\eps_2) = \msovarphi{psuc}(\eps_1,\eps_2)\ \wedge\ {\msovarphi{TEintLIFE}}_\delta(\eps_1,\eps_2)\ \wedge\ \neq {\msovarphi{TEintLIFE}}_\delta(\eps_2,\eps_1)
    }
    {FO}{folklore}
    
\MSOFormulaD{deltapsucIM}{possible $\delta$-successor (\vim, \tim)}
    {an ordered pair of temporal edges could be consecutive in a temporal path with waiting time at most $\delta$}
    {\delta}
    {\vim, \tim}{constant$^*$}
    {(\eps_1,\eps_2) = \msovarphi{psuc}(\eps_1,\eps_2)\ \wedge\ {\msovarphi{TEintIM}}_\delta(\eps_1,\eps_2)\ \wedge\ \neq {\msovarphi{TEintIM}}_\delta(\eps_2,\eps_1)
    }
    {FO$^*$}{folklore}

\msoalias{deltapsuc}{deltapsucLIFE}
\msoalias{deltapsuc}{deltapsucIM}

\MSOFormulaD{restlesspath}{restless temporal path}
    {there exists a temporal path from $u$ to $v$ whose waiting time between consecutive edges is bounded by $\delta$}
    {\delta}
    {\lifetime, \tdegree, \vim, \tim}{constant$^*$}
    {(u,v) = 
        \exists P_\ecal\subseteq \ecal\ \exists \varepsilon_s,\varepsilon_f\in P_\ecal\colon \Big(
            \inc(\varepsilon_s,u)\ \wedge\ \inc(\varepsilon_f,v)\ \wedge
            \msovarphi{deg=1}(u,P_\ecal)\ \wedge\\& \msovarphi{deg=1}(v,P_\ecal)\ \wedge\ 
            \forall \varepsilon\in P_\ecal\setminus\{\varepsilon_s,\varepsilon_f\}, \forall w\in V\colon
                \big(\inc(\varepsilon,w)\rightarrow \msovarphi{deg=2}(w,P_\ecal)\big)\ \wedge\\&
            \forall \varepsilon\in P_\ecal\setminus\{\varepsilon_f\}\ \exists \varepsilon'\in P_\ecal\colon
                \msovarphi{deltapsuc}_\delta(\varepsilon,\varepsilon')
        \Big)
    }
    {MSO}{folklore}

%% file: Formulas/ReachProblems/subsec--RP-TCC.tex
Temporal connected components generalize the notion of static connected components to temporal graphs by requiring mutual temporal reachability between vertices. 

\textsc{Open Temporal Connected Components} by~\cite{bhadra_complexity_2003} require that for every pair of vertices $u,v$ in the component there exists a temporal path from $u$ to $v$ and vice versa, without restricting paths to remain inside the component.
\MSOFormula{opentcc}{open temporal connected component \cite{deligkas_parameterized_2025}}
    {between every pair of vertices $u,v$ in $X$ there exists a temporal path from $u$ to $v$ and from $v$ to $u$, and $X$ is maximal}
    {X}
    {\lifetime, \tdegree, \vim, \tim}{constant}
    {
        \forall u,v\in X \colon \msovarphi{path}(u,v) \; \wedge\\
        & \neg\exists y\in V\colon \Big( y\notin X \wedge \forall u\in X \colon\big( \msovarphi{path}(u,y) \wedge \msovarphi{path}(y,u) \big)\Big)
    }
    {FO$^*$}{\cite{deligkas_parameterized_2025}}

\textsc{Closed Temporal Connected Components} by~\cite{bhadra_complexity_2003} strengthen the requirement of open components by additionally demanding that all witnessing temporal paths use only vertices from the component itself. As a consequence, closed temporal connected components are not hereditary: a subset of a closed component does not necessarily form a closed component. This leads to a different formulation of maximality.
\MSOFormula{closedtcc}{closed temporal connected component \cite{deligkas_parameterized_2025}}
    {between every pair of vertices $u,v$ in $X$ there exists a temporal path from $u$ to $v$ and from $v$ to $u$ using only vertices in $X$, and $X$ is maximal}
    {X}
    {\lifetime, \tdegree, \vim, \tim}{constant}
    {
        \forall u,v\in X \colon \msovarphi{pathV}(u,v,X) \wedge \neg\exists Y\subseteq V \colon \Big( X\subseteq Y\; \wedge\\&
        \forall u,v\in Y \big( \msovarphi{pathV}(u,v,Y) \wedge \msovarphi{pathV}(v,u,Y) \big)\Big)
    }
    {MSO}{\cite{deligkas_parameterized_2025}}
\noindent Removing the maximality conditions from  \msovarphi{opentcc} and \msovarphi{closedtcc} yields formulas expressing open and closed temporally connected sets, respectively.

\textsc{Unilateral Open} and \textsc{Unilateral Closed Temporal Connected Components} by~\cite{costa_computing_2023} only require reachability in one direction for each pair of vertices.
\MSOFormula{Uopentcc}{unilateral open temporal connected component}
    {between every pair of vertices $u,v$ in $X$ exists a temporal path from $u$ to $v$ or from $v$ to $u$ using vertices in $V$, and $X$ is maximal}
    {X}
    {\lifetime, \tdegree, \vim, \tim}{constant}
    {
        \forall u,v\in X\colon \big(\msovarphi{path}(u,v) \vee \msovarphi{path}(v,u) \big) \; \wedge\\&
        \neg\exists y\in V \colon \Big( y\notin X \wedge \forall u\in X \colon \big( \msovarphi{path}(u,y) \vee \msovarphi{path}(y,u) \big)\Big)
    }
    {FO$^*$}{novel}

\MSOFormula{Uclosedtcc}{unilateral closed temporal connected component}
    {between every pair of vertices $u,v$ in $X$ exists a temporal path from $u$ to $v$ or from $v$ to $u$ using vertices in $X$, and $X$ is maximal.}
    {X}
    {\lifetime, \tdegree, \vim, \tim}{constant}
    {
        \forall u,v\in X\colon \big(\msovarphi{pathV}(u,v,X) \vee \msovarphi{pathV}(v,u,X) \big) \; \wedge\\&
        \neg\exists Y\subseteq V \colon \Big( X\subseteq Y \wedge \forall u,v\in Y \colon\big( \msovarphi{pathV}(u,v,Y) \vee \msovarphi{pathV}(v,u,Y) \big)\Big)
    }
    {MSO}{novel}

%% file: Formulas/ReachProblems/subsec--RP-separators.tex
\textsc{Temporal Separators} and \textsc{Temporal Cuts} (see \cite{berman_vulnerability_1996,kempe_connectivity_2002,fluschnik_temporal_2020}) generalize static separation concepts to the temporal setting by removing elements in order to remove temporal reachability.
Depending on whether vertices, static edges, or temporal edges are removed, one obtains different notions of temporal separators.

\textsc{Temporal $(s,z)$-Separation} by~\cite{berman_vulnerability_1996,kempe_connectivity_2002} separates $s$ from $z$ by removing vertices from the graph. \cite{zschoche_complexity_2020} provided an FO formula for \textsc{Temporal $(s,z)$-Separation} on strict temporal graphs in the lifetime encoding.
\MSOFormula{Vsep}{temporal vertex separator \cite{zschoche_complexity_2020}}
    {removing the vertices in $X$ removes all temporal paths from $s$ to $z$}
    {X,s,z}
    {\lifetime, \tdegree, \vim, \tim}{constant}
    {
        \neg \msovarphi{pathV}(s,t,V\setminus X)
    }
    {FO$^*$}{\cite{zschoche_complexity_2020},novel}

\MSOFormula{SEsep}{temporal static-edge separator}
    {removing the static edges in $E'$ removes all temporal paths from $s$ to $z$}
    {E',s,z}
    {\lifetime, \tdegree, \vim, \tim}{constant}
    {
        \neg \msovarphi{pathSE}(s,t,E\setminus E')
    }
    {FO$^*$}{novel}

A temporal separator which separates $s$ from $z$ by removing temporal edges from the graph was studied by~\cite{berman_vulnerability_1996} in the context of computing temporal-edge-disjoint paths and (dis)proving temporal analogues of Menger's Theorem.
\MSOFormula{TEsep}{temporal temporal-edge separator}
    {removing the temporal edges in $\mathcal{E}'$ removes all temporal paths from $s$ to $z$}
    {\mathcal{E}',s,z}
    {\lifetime, \tdegree, \vim, \tim}{constant}
    {
        \neg \msovarphi{pathTE}(s,t,\mathcal{E}\setminus\mathcal{E}')
    }
    {FO$^*$}{novel}

%% file: Formulas/ReachProblems/subsec--RP-Spanner.tex
\textsc{Temporal Spanners} (see \cite{casteigts_temporal_2021,kempe_connectivity_2002}) are sparse subgraphs of a temporal graph that preserve temporal reachability between all pairs of vertices. They are the temporal analogue of spanning trees in static graphs.
\MSOFormula{tspanner}{temporal spanner}
    {$X$ preserves temporal reachability between all pairs of vertices}
    {X}
    {\lifetime, \tdegree, \vim, \tim}{constant}
    {
        \forall u,v\in V \colon
        \big(
            \msovarphi{path}(u,v)
            \rightarrow
            \msovarphi{pathV}(u,v,X)
        \big)
    }
    {FO$^*$}{novel}

%% file: Formulas/ReachProblems/subsec--RP-Exploration.tex
\textsc{Temporal Exploration} (see \cite{aceto_how_2008,erlebach_temporal_2021}) is defined by the existence of a temporal walk that traverses all vertices or all static edges of the temporal graph.
\MSOFormula{SEexplorLIFE}{static edge exploration (\lifetime)}
    {there exists a temporal walk traversing all static edges}
    {}
    {\lifetime}
    {\lifetime}
    {\exists x_0,\ldots x_{\lifetime}\in V\colon \Big( \ 
    \msovarphi{VwalkLIFE}(x_0,\ldots,x_{\lifetime})\ \wedge\\&
    \forall e\in E \colon \bigvee_{0\leq i<\lifetime} \big( x_i\neq x_{i+1} \wedge \inc(e,x_i) \wedge \inc(e,x_{i+1})
    \big)\ \Big)
    }
    {FO$^*$}{novel}

\MSOFormula{VexplorLIFE}{vertex exploration (\lifetime)}
    {there exists a walk traversing all vertices}
    {}
    {\lifetime}
    {\lifetime}
    {\exists x_0,\ldots x_{\lifetime}\in V\colon \Big(\ 
    \msovarphi{VwalkLIFE}(x_0,\ldots,x_{\lifetime})\ \wedge \\&
    \forall v\in V \colon \bigvee_{0\leq i\leq\lifetime} v=x_i
    \ \Big)
    }
    {FO$^*$}{novel}




%% file: Formulas/subsec--Clique-IS.tex
\textsc{Temporal Clique} was defined by \cite{viard_computing_2016}.
The problem asks for a vertex set $X_V$ such that, for every pair of distinct vertices $u,v\in X_V$ and for every time window of length $\Delta$, there is a temporal edge between $u$ and $v$ within that window.
\MSOFormulaD{cliqueLIFE}{$\Delta$-temporal clique (\lifetime)}
    {there exists a set $X_V$ of vertices which is a $\Delta$ temporal clique.}
    {\Delta}
    {\lifetime}
    {constant$^*$}
    {() = \exists X_V\subseteq V\colon\Big(\ 
    \forall u \in X_V, \forall v \in X_V\setminus\{u\} \colon
    \\\Big(&
        \forall t \in L, \exists \eps\in\ecal, \exists t' \in L \colon
        \big(
            {\msovarphi{LintLIFE}}_{\Delta}(t,t') \wedge \inc(\eps,u) \wedge \inc(\eps,v) \wedge \pres(\eps,t')
        \big)
    \Big)\ \Big)
    }
    {MSO}
    {novel}

    \MSOFormulaD{cliqueVIM}{$\Delta$-temporal clique (\vim)}
    {there exists a set $X_V$ of vertices which is a $\Delta$ temporal clique.}
    {\Delta}
    {\vim}
    {constant$^*$}
    {() = \exists X_V\subseteq V\colon\Big(\ 
    \forall u \in X_V, \forall v \in X_V\setminus\{u\} \colon
    \\\Big(&
        \forall \vtimbag \in \vtimbagset, \exists \eps\in\ecal, \exists \vtimbag' \in \vtimbagset \colon
        \big(
            {\msovarphi{BintIM}}_{\Delta}(\vtimbag,\vtimbag') \wedge \inc(\eps,u) \wedge \inc(\eps,v) \wedge \pres(\eps,\vtimbag')
        \big)
    \Big)\ \Big)
    }
    {MSO}
    {novel}

The corresponding temporal analogue of \textsc{Independent Set} was defined by \cite{hermelin_temporal_2022}, where for every pair of distinct vertices $u,v\in X_V$ and for every time window of length $\Delta$, there is no temporal edge between $u$ and $v$ within that window.

\MSOFormulaD{independentL}{$\Delta$-temporal independent set (\lifetime)}
    {there exists a set $X_V$ of vertices is a $\Delta$ temporal independent set.}
    {\Delta}
    {\lifetime}
    {constant$^*$}
    {() = \exists X_V\subseteq V\colon \Big(\ 
    \forall u \in X_V, \forall v \in X_V\setminus\{u\} \colon
    \\\Big(&
        \forall t, \in L, \exists \eps\in\ecal, \forall t'\in L \colon
        \big(
            {\msovarphi{LintLIFE}}_{\Delta}(t,t') \wedge
            \inc(\eps,u) \wedge \inc(\eps,v)
            \implies \neg \pres(\eps,t')
        \big)
    \Big)\ \Big)
    }
    {MSO}
    {novel}

    \MSOFormulaD{independentVIM}{$\Delta$-temporal independent set (\vim)}
    {there exists a set $X_V$ of vertices is a $\Delta$ temporal independent set.}
    {\Delta}
    {\vim}
    {constant$^*$}
    {() = \exists X_V\subseteq V\colon \Big(\ 
    \forall u \in X_V, \forall v \in X_V\setminus\{u\} \colon
    \\\Big(&
        \forall \vtimbag, \in \vtimbagset, \exists \eps\in\ecal, \forall \vtimbag'\in \vtimbagset \colon
        \big(
            {\msovarphi{BintIM}}_{\Delta}(\vtimbag,\vtimbag') \wedge
            \inc(\eps,u) \wedge \inc(\eps,v)
            \implies \neg \pres(\eps,\vtimbag')
        \big)
    \Big)\ \Big)
    }
    {MSO}
    {novel}

%% file: Formulas/subsec--FES_EdgeCover.tex
\textsc{Temporal Feedback Edge Set} by~\cite[Definition 3]{haag_feedback_2022} is a temporal analogue of the classical \textsc{Feedback Edge Set} problem. Note that the following formulae hold becuase the temporal path formula nested within them requires that the set of temporal edges forming the path is non-empty, thus it is a non-trivial path.
\MSOFormula{tfTEs}{temporal feedback temporal edge set \cite[Thm. 23]{haag_feedback_2022}}
    {temporal edge set $X_\ecal\subseteq \ecal$ such that after removing $X$, there exists no temporal cycle: temporal walk over $v_1,\dots,v_\ell$ with $v_1,\dots,v_{\ell-1}$ pairwise distinct and $v_1=v_\ell$} 
    {X_\ecal}
    {\lifetime, \tdegree, \vim, \tim}
    {constant}
    {
    \forall v \in V\colon \neg\msovarphi{pathTE}(v,v,\ecal\setminus X_\ecal)
    }
    {FO$^\star$}
    {\cite[Thm. 23]{haag_feedback_2022}}

\textsc{Temporal Feedback Connection Set} is the static edge variant of \textsc{Temporal Feedback Edge Set}.
\MSOFormula{tfSEs}{temporal feedback static edge set \cite[Thm. 23]{haag_feedback_2022}}
    {static edge set $X_E\subseteq E$ such that after removing $X$, there exists no temporal cycle: temporal walk over $v_1,\dots,v_{\ell}$ with $v_1,\dots,v_{\ell-1}$ pairwise distinct and $v_1=v_\ell$}
    {X_E}
    {\lifetime, \tdegree, \vim, \tim}
    {constant}{
        \forall v \in V \colon \neg \msovarphi{pathSE}(v,v,\ecal \setminus X_E)
    }
    {FO$^\star$}
    {\cite[Thm. 23]{haag_feedback_2022}}

%% file: Formulas/subsec--Coloring.tex
\textsc{Temporal $k$-Colouring} by~\cite{mertzios_sliding_2019} assigns to every vertex and every time step one of $k$ colours such that, at each time step, no temporal edge connects two vertices of the same colour. In contrast to static colouring, colours may change over time.
\MSOFormulaD{colour}{temporal $k$-colouring (\lifetime)}
    {there exists a temporal $k$-colouring of the input temporal graph}
    {k}
    {\lifetime}
    {$k\cdot\lifetime$}
    {() = 
    \exists \{Y^t_c \subseteq V \mid t\in[\lifetime],\ c\in[k]\} \colon \\
    \Big(
        &\forall t \in L, \forall v \in V \colon
        \big(
            \bigvee_{c\in[k]} \big( v \in Y^t_c \wedge \bigwedge_{c'\in[k]}
                ( v \in Y^t_{c'} \implies c=c')\big)\ 
        \big)
        \ \wedge\\&
        \forall t \in L, \forall e \in E, \forall u,v \in V\colon \bigwedge_{c \in [k]}
        \big(
            \inc(e,u)\wedge \inc(e,v)\wedge \pres(e,t)\wedge u\in Y^t_c
            \implies v\notin Y^t_c
        \big) 
    \ \Big)
    }
    {MSO}
    {novel}


%% file: Formulas/subsec--Matching.tex
A \textsc{$\Delta$-Temporal Matching} as defined by \cite{mertzios_computing_2020} is a set of temporal edges $M_\ecal\subseteq\ecal$ such that for every vertex $v\in V$ the times at which $v$ is incident to edges in $M_\ecal$ differ pairwise by at least $\Delta$. The problem \textsc{0-1-timed Matching} from \cite{mandal_maximum_2022} asks for a $\Delta$-temporal matching with $\Delta=2$.


\MSOFormulaD{deltamatchingLIFE}{$\Delta$-temporal matching (\lifetime)}
    {for all $v\in V$ and every pair of distinct temporal edges $(e,t),(e',t')\in M_\ecal$ incident to $v$ holds $\lvert t - t'\rvert\geq \Delta$}
    {\Delta}
    {\lifetime
    }
    {constant$^*$}
    {() = \exists M_\ecal\subseteq\ecal\colon\Big(\ 
    \forall v\in V, \forall \varepsilon_1 \in M_\ecal, \varepsilon_2\in M_\ecal\setminus\{\varepsilon_1\} \colon \\& 
    \big(\ \inc(\varepsilon_1,v)\wedge \inc(\varepsilon_2,v) \implies 
            ( \neg {\msovarphi{TEint}}_{\Delta}(\varepsilon_1,\varepsilon_2)\ \wedge\ \neg {\msovarphi{TEint}}_{\Delta}(\varepsilon_2,\varepsilon_1) )\ 
    \big)
    \ \Big)
    }
    {MSO}{novel}

    \MSOFormulaD{deltamatchingVIM}{$\Delta$-temporal matching (\vim)}
    {for all $v\in V$ and every pair of distinct temporal edges $(e,t),(e',t')\in M_\ecal$ incident to $v$ holds $\lvert t - t'\rvert\geq \Delta$}
    {\Delta}
    {\vim
    }
    {constant$^*$}
    {() = \exists M_\ecal\subseteq\ecal\colon\Big(\ 
    \forall v\in V, \forall \varepsilon_1 \in M_\ecal, \varepsilon_2\in M_\ecal\setminus\{\varepsilon_1\}, \exists \vtimbag_1,\vtimbag_2\in \vtimbagset \colon \\& 
    \big(\ (\pres(\eps_1,\vtimbag_1) \wedge \pres(\eps_2,\vtimbag_2)) \wedge \inc(\varepsilon_1,v)\wedge \inc(\varepsilon_2,v) \implies 
            ( \neg {\msovarphi{TEint}}_{\Delta}(\vtimbag,\vtimbag)\ \wedge\ \neg {\msovarphi{TEint}}_{\Delta}(\vtimbag_2,\vtimbag_1) )\ 
    \big)
    \ \Big)
    }
    {MSO}{novel}

\textsc{Temporal Matching} by \cite{cioni_matching_2025} is defined as a set of temporal edges such that no two temporal edges in that matching may be consecutive in a temporal path.
\MSOFormula{tmatching}{temporal matching}
    {for every pair of distinct temporal edges $\varepsilon_1,\varepsilon_2\in M_\ecal$, $\varepsilon_2$ cannot follow $\varepsilon_1$ in a temporal path}
    {M_\ecal}
    {\lifetime, \tdegree, \vim, \tim}
    {constant$^*$}
    {
    \forall \varepsilon_1 \in M_\ecal, \forall \varepsilon_2 \in M_\ecal\setminus\{\varepsilon_1\} \colon
    \big(
        \neg \msovarphi{psuc}(\varepsilon_1,\varepsilon_2)
        \ \wedge\
        \neg \msovarphi{psuc}(\varepsilon_2,\varepsilon_1)
    \big)
    }
    {FO$^\star$}
    {novel}

%% file: Formulas/subsec--Cover.tex
Covering problems in temporal graphs can be classified by the type of objects being covered and by the temporal semantics under which coverage is required, such as snapshot-wise, over time, interval-based, or with bounded changes.
The objects covering/being covered are the static edges~$E$, the temporal edges~\ecal, and the vertices~$V$ or temporal vertex instances~$V\times[\Lambda]$. In this subsection the elements in $V$ are referred to as \textit{static vertices}, and in $V\times[\Lambda]$ as \textit{temporal vertices}.

\subsubsection{Covering temporal edges}
All problems in this subsection require that every temporal edge $(e,t)\in\ecal$ is covered by an incident object at the time~$t$ when it appears. The presented variants range from path-based coverings to snapshot-wise and interval-based vertex covers.

\input{Formulas/VCandDS/f-EECpath}

\input{Formulas/VCandDS/f-VCmultistage}

\input{Formulas/VCandDS/f-VCtimeline}

\subsubsection{Covering static edges}
All problems in this subsection require that every static edge $e\in E$ is covered by an incident object at some time when the edge is present. The presented variants range from vertex covers with coverage over time, sliding windows, and bounded changes between snapshots.

\input{Formulas/VCandDS/f-VCtemporal}

\input{Formulas/VCandDS/f-VCtemporalDelta}

\subsubsection{Covering temporal vertices}
The problem in this subsection requires that every temporal vertex $(v,t)\in V\times[\Lambda]$ that is incident to at least one temporal edge is covered by any incident object at time~$t$.

\input{Formulas/VCandDS/f-ECtemporal}

\subsubsection{Covering static vertices}
All problems in this subsection require that every static vertex $v\in V$ is covered by a selected object according to a specified temporal semantics. The presented variants range from snapshot-wise domination to domination over time, interval-based activation, reachability-based coverage, and bounded-change constraints.

\input{Formulas/VCandDS/f-DSsnapshot}

\input{Formulas/VCandDS/f-DStimeline}

\input{Formulas/VCandDS/f-DSover_time}

\input{Formulas/VCandDS/f-DSpermanent}

\input{Formulas/VCandDS/f-DSreach}

%% file: Formulas/VCandDS/f-EECpath.tex
\textsc{Temporal Path Exact Edge Cover} by~\cite{deligkas_how_2025} covers \textit{temporal edges} using \textit{temporal paths}.
    \MSOFormula{TPcoverTEE}{temporal path exact edge cover}
        {there exists a family of temporal edge sets $P_1,\dots,P_k$ such that each $P_i$ is a temporal path and for every $\eps\in\ecal$ exists exactly one $i\in[k]$ such that $\eps\in P_i$}
        {}
        {\lifetime, \tdegree, \vim, \tim}
        {$k^2$}
        {
        \exists P_1,\dots,P_k\subseteq \ecal\colon\Big(\ 
            \bigwedge_{i\in[k]} \big(\exists u,v\in V\colon \msovarphi{PEpath}(u,v,P_i) \big)\ \wedge \\&
            \forall \varepsilon \in \ecal \colon \Big(
            \bigvee_{i\in[k]} \big(\varepsilon \in P_i \wedge \bigwedge_{j\in[k], j\neq i} \varepsilon\notin P_j\big)\Big)
        \ \Big)}
        {MSO}
        {novel}

%% file: Formulas/VCandDS/f-VCmultistage.tex
\textsc{Multistage Vertex Cover} by~\cite{fluschnik2022multistage} covers \textit{temporal edges} using \textit{temporal vertices}. 
    It asks for one vertex cover $X_t$ for each snapshot $G_t$, and the sets are only allowed to change by at most $l$ vertices when moving to the next snapshot.
    \MSOFormula{multistageTVcoverTE}{multistage vertex cover}
        {there exists a family of vertex sets $X_1,\ldots,X_{\lifetime}\subseteq V$ such that 
        for every $t\in[\lifetime]$, $X_t$ is a vertex cover of $G_t$, and for every $t\in[\lifetime-1]$ it holds that $\lvert (X_t \setminus X_{t+1}) \cup (X_{t+1} \setminus X_t)\rvert \le \ell$}
        {}
        {\lifetime}
        {$\lifetime^*$}
        {\exists X_1,\ldots,X_{\lifetime} \subseteq V\bigwedge_{t\in[\lifetime]} \forall \eps\in\ecal \colon \big( \pres(\eps,t) \implies \exists v\in V\colon (\inc(v,\eps) \wedge v\in X_t \big)\ \wedge \\&
        \bigwedge_{t\in [\lifetime-1]} \exists S\subseteq V\colon \Big( \neg{\msovarphi{card}}_{\ell+1}(S) \wedge 
        \forall v\in V \colon  v\in S \biimplies \big( (v\in X_t \wedge v\notin X_{t+1}) \vee  (v\notin X_t \wedge v\in X_{t+1}) \big) \Big)
        }{MSO}{novel}

%% file: Formulas/VCandDS/f-VCtimeline.tex
\textsc{Timeline Vertex Cover} by~\cite{rozenshtein_network-untangling_2021} covers \textit{temporal edges} using \textit{temporal vertices over intervals}.
    Rather than selecting a vertex separately at each time step, choosing a vertex at time $t$ activates it for the entire interval $[t,t+\ell-1]$. A temporal edge must be covered at the time it appears by an endpoint that is active in an interval containing this time.
    \MSOFormula{kltimelineTVcoverSE}{timeline vertex cover}
        {there exists a family of vertex sets $X_1,\ldots,X_{\lifetime-\ell+1}\subseteq V$ such that 
        for every temporal edge $(\eps,t)$ there exists a time $s\in[\lifetime-\ell+1]$ and a vertex $v\in V$ such that $t\in[s,s+\ell-1]$, $\inc(\eps,v)$, and $v\in X_s$}
        {}
        {\lifetime}
        {$\lifetime^3$}
        {\exists X_1,\ldots,X_{\lifetime-\ell} : 
        \Big(
        \forall \eps\in \ecal, \forall t\in[\lifetime]\colon\\&
            \pres(\eps,t) \implies \big(
            \exists v\in V\colon \inc(\eps,v) \wedge \bigvee_{s\in[\lifetime-\ell+1]} 
                \big(
                    v\in X_s  \wedge
                    {\msovarphi{LintLIFE}}_{\ell}(t,s)
                \big)
            \big)
        \Big) \ \wedge\\&
        \Big(
        \forall v\in V\colon
            \neg\bigvee_{t_1,\ldots,t_{k+1}\in[\lifetime-\ell+1]} \big(
                \bigwedge_{i,j\in[k+1], i\neq j}(t_i\neq t_j)\ \wedge\
                \bigwedge_{i=1}^{k+1}(v\in X_{t_i}) \big)
        \Big)
        }
        {MSO}{novel}

        Note here that we use \lifetime to bound $k$, since we can assume that each vertex is in at most \lifetime many different intervals.

%% file: Formulas/VCandDS/f-VCtemporal.tex
\textsc{Temporal Vertex Cover} by~\cite{akrida_temporal_2020} covers \textit{static edges} using \textit{temporal vertex instances}.
        \MSOFormula{SVcoverSE}{temporal vertex cover (\lifetime)}
        {there exists a family of vertex sets $X_1,\ldots,X_{\lifetime}$ such that for every static edge $e=uv\in E$ there exists a time $t\in[\lifetime]$ such that $(e,t)\in\ecal$, and $u\in X_t$ or $v\in X_t$}
        {}
        {\lifetime}
        {\lifetime}
        {\exists X_1,\ldots,X_{\lifetime}\subseteq V\colon \Big(\ 
            \forall e\in E, \exists t\in [\lifetime]\colon \big(
            \pres(e,t) \wedge \exists v\in V\colon(\inc(e,v) \wedge v\in X_t)
            \big)
        \ \Big)}
        {MSO}{novel}

%% file: Formulas/VCandDS/f-VCtemporalDelta.tex
\textsc{$\Delta$-Temporal Vertex Cover} by~\cite{akrida_temporal_2020} covers \textit{static edges} using \textit{temporal vertex instances}.
    It is the sliding window variant of \textsc{Temporal Vertex Cover} with the additional criterion that every static edge is covered by a vertex-time pair in every interval of length $\Delta$. 
    \MSOFormulaD{DeltaTVcoverSE}{$\Delta$-temporal vertex cover (\lifetime)}
        {there exists a family of vertex sets $X_1,\ldots,X_{\lifetime}$ such that for every static edge $e=uv\in E$ and every time window beginning $t\in[\lifetime]$, there exists a time $t'\in[t,t+\Delta]$ such that $(e,t')\in\ecal$, and $u\in X_{t'}$ or $v\in X_{t'}$}
        {\Delta}
        {\lifetime}
        {$\lifetime^*$}
        {() = \exists X_1,\ldots,X_{\lifetime}\subseteq V\colon\Big(\ 
            \forall e\in E, \forall t\in[\lifetime], \exists t'\in [\lifetime]\colon \big( \\&
            t \timebefore t'\ \wedge {\msovarphi{LintLIFE}}_\Delta(t',t) \wedge \pres(e,t') \wedge \exists v\in V\colon(\inc(e,v) \wedge v\in X_{t'})
            \big)
        \ \Big)}
        {MSO}{novel}

%% file: Formulas/VCandDS/f-ECtemporal.tex
\textsc{Temporal Edge Cover} by~\cite{cioni_matching_2025} covers \emph{temporal vertices} by \textit{static edges}. A temporal edge cover is a set of temporal edges such that every vertex--time pair $(v,t)$ that participates in at least one temporal edge is incident to at least one selected temporal edge at time~$t$.
    \MSOFormula{SEcoverTV}{temporal edge cover}
        {there exists a set of static edges $X_E\subseteq E$ such that every non-isolated vertex in every snapshot is incident to an active edge from $X_E$}
        {}
        {\lifetime}
        {constant}
        {
        \exists X_E\subseteq E, \forall t\in[\lifetime], \forall v\in V\colon\\&
        \Big(
            \big(\exists \eps\in \ecal\colon (\inc(\eps,v)\wedge \pres(\eps,t))\big)
             \implies
            \exists e'\in X_E\colon (\inc(e',v)\wedge \pres(e',t))
        \Big)
        }
        {MSO}{novel}

%% file: Formulas/VCandDS/f-DSsnapshot.tex
\textsc{Snapshot Dominating Set} (also called \textit{temporal} or \textit{evolving dominating set})
by~\cite{casteigts_journey_2018} covers \textit{static vertices} using \textit{temporal vertices} - it requires \textit{coverage in each snapshot}.
\MSOFormula{snapshotTVcoverSV}{snapshot dominating set (\lifetime)}
    {there exists a family of static vertex sets $X_1,\ldots,X_{\lifetime}$ such that for every $t\in[\lifetime]$, $X_t$ is a dominating set of $G_t$}
    {}
    {\lifetime}
    {\lifetime}
    {
        \exists X_1,\ldots,X_{\lifetime}\subseteq V\colon
        \bigwedge_{t\in[\lifetime]}\ \forall v\in V\colon
        \Big(
            v\in X_t\ \vee\\&
            \exists u\in X_t, \exists \eps\in\ecal\colon
            \big(\pres(\eps,t)\wedge \inc(\eps,u)\wedge \inc(\eps,v)\big)
        \Big)
    }
    {MSO}{novel}

%% file: Formulas/VCandDS/f-DStimeline.tex
Like \textsc{Timeline Vertex Cover}, \textsc{Timeline Dominating Set} by~\cite{herrmann_timeline_2025} covers \textit{temporal vertices} using \textit{temporal vertices over time intervals}.
Rather than selecting a vertex separately at each time step, choosing a vertex at time $t$ activates it for the entire interval $[t,t+\ell-1]$. Each vertex must be dominated by being in in the dominating set, or being the endpoint of an edge which is active in an interval in which the other endpoint is selected.
    \MSOFormulaD{timelineTVcoverTV}{timeline dominating set}
        {there exists vertex sets $X_1,\ldots,X_{\lifetime-\ell+1}$ such that every vertex is dominated in every snapshot by a vertex active in an interval of length $\ell$, and each vertex is active in at most $k$ such intervals}
        {k,l}
        {\lifetime}
        {$\lifetime^*$}
        {() = \exists X_1,\ldots,X_{\lifetime-\ell+1}\subseteq V, \forall v\in V \colon \\&
        \Bigg(\ 
        \neg\Big(
            \bigvee_{t_1,\ldots,t_{k+1}\in[\lifetime-\ell+1]} \big(
                \bigwedge_{i,j\in[k+1], i\neq j}(t_i\neq t_j)\ \wedge\
                \bigwedge_{i=1}^{k+1}(v\in X_{t_i})
            \big)
        \Big) \ \wedge\ 
        \Big(
        \forall t\in[\lifetime], \exists s\in[\lifetime-\ell+1]\colon\\&
            \Big(\ {\msovarphi{LintLIFE}}_\ell(t,s)\ \wedge
            \big(
                v\in X_s\ \vee\
                \exists u\in V,\ \exists e\in E\colon
                    \big(u\in X_s \wedge \inc(e,u)\wedge \inc(e,v)\wedge \pres(e,t)\big)
            \big)\  \Big) 
        \Big)
        \ \Bigg)
        }
        {MSO}{novel}

%% file: Formulas/VCandDS/f-DSover_time.tex
\textsc{Temporal Dominating Set over time} (also called \emph{temporal domination}) by \cite{casteigts_journey_2018} covers \textit{static vertices} using \textit{static vertices}. It asks for a vertex set $S\subseteq V$ such that every vertex outside $S$ is dominated at least once over time. This is equivalent to solving \textsc{Dominating Set} on the footprint graph.

\MSOFormula{overtimeSVcoverSV}{temporal dominating set over time}
    {there exists a vertex set $S\subseteq V$ such that every vertex $v\notin S$ has a neighbor in $S$ at least once over time}
    {}
    {\lifetime}{constant}
    {
        \exists S\subseteq V, \forall v\in V\colon
        \Big(
            v\in S\ \vee\
            \exists u\in S, \exists t\in[\lifetime], \exists \eps\in\ecal\colon
                \big(\pres(\eps,t)\ \wedge\\& \inc(\eps,u)\wedge \inc(\eps,v)\big)
        \Big)
    }
    {MSO}{novel}

%% file: Formulas/VCandDS/f-DSpermanent.tex
\textsc{Permanent Dominating Set} by~\cite{casteigts_journey_2018} covers
\textit{static vertices} using \textit{static vertices}.
It asks for a vertex set $S\subseteq V$ that dominates every snapshot $G_t$.
\MSOFormula{permanentSVcoverSV}{permanent dominating set}
    {there exists a vertex set $S\subseteq V$ such that for every $t\in[\lifetime]$, $S$ is a dominating set of $G_t$}
    {}
    {\lifetime}{constant}
    {
        \exists S\subseteq V, \forall t\in[\lifetime], \forall v\in V\colon
        \Big(
            v\in S\ \vee\
            \exists u\in S, \exists \eps\in\ecal\colon
            \big(\pres(\eps,t)\ \wedge\\& \inc(\eps,u)\wedge \inc(\eps,v)\big)
        \Big)
    }
    {MSO}{}

%% file: Formulas/VCandDS/f-DSreach.tex
\textsc{Temporal Reachability Dominating Set} by~\cite{kutner_temporal_2023}
covers \textit{static vertices} using \textit{static vertices}, where coverage is achieved by temporal reachability rather than by incident edges.
\cite{kutner_temporal_2023} provide an MSO formulation for the \lifetime\ encoding;
this can be extended to all encodings as follows.
\MSOFormula{tempreachSVcoverSV}{temporal reachability dominating set \cite{kutner_temporal_2023}}
    {there exists a vertex set $S$ such that every vertex $v\in V$ is either in $S$ or reachable by a temporal path from some vertex in $S$}
    {}
    {\lifetime, \tdegree, \vim, \tim}
    {constant}
    {
        \exists S\subseteq V, \forall v\in V\colon
        \Big(
            v\in S \ \vee\ \exists u\in S \colon \msovarphi{path}(u,v)
        \Big)
    }
    {MSO}{\cite{kutner_temporal_2023}}